\documentclass[11pt, letterpaper]{article}

\usepackage[english]{babel}
\usepackage[utf8]{inputenc}
\usepackage[T1]{fontenc}
\usepackage{microtype}

\usepackage[margin=1in]{geometry}

\usepackage[numbers]{natbib}

\usepackage[dvipsnames]{xcolor}
\usepackage{xfrac}
\usepackage{float}

\usepackage{todonotes}

\usepackage{mathtools}
\usepackage{amsthm, amsfonts, thmtools}
\declaretheorem{theorem}
\declaretheorem[sibling=theorem]{lemma}
\declaretheorem[sibling=theorem]{corollary}
\declaretheorem[style=definition]{definition}

\usepackage{algorithm, algcompatible}

\usepackage{graphicx}
\usepackage{subcaption}

\usepackage[hidelinks]{hyperref}
\usepackage{xcolor}
\newcommand{\declarecolor}[2]{\definecolor{#1}{RGB}{#2}\expandafter\newcommand\csname #1\endcsname[1]{\textcolor{#1}{##1}}}
\declarecolor{Green}{112, 173, 71}
\declarecolor{Navy}{68, 114, 196}
\hypersetup{
	colorlinks=true,
	pdfpagemode=UseNone,
	citecolor=Green,
	linkcolor=Navy,
	urlcolor=Navy,
	pdfstartview=FitW,
}

\title{Edge-Weighted Online Bipartite Matching\footnote{
This paper merges and refines the results in
\href{https://arxiv.org/abs/1704.05384}{arXiv:1704.05384v2},
\href{https://arxiv.org/abs/1910.02569}{arXiv:1910.02569},
and
\href{https://arxiv.org/abs/1910.03287}{arXiv:1910.03287}.
In particular,
we fix a bug in
\href{https://arxiv.org/abs/1910.03287}{arXiv:1910.03287}
and have a smaller competitive ratio as a result. 
Appendix~\ref{sec:origin} discusses the
connections between the primal-dual algorithm in this work
and the original algorithm of Fahrbach and Zadimoghaddam.
}}

\author{
    Matthew Fahrbach%
    \thanks{Google Research.
    Email: \href{mailto:fahrbach@google.com}{fahrbach@google.com}.}
    \and
    Zhiyi Huang%
    \thanks{The University of Hong Kong.
    Email: \href{mailto:zhiyi@cs.hku.hk}{zhiyi@cs.hku.hk}.}
    \and
    Runzhou Tao%
    \thanks{Columbia University.
    Email: \href{mailto:runzhou.tao@columbia.edu}{runzhou.tao@columbia.edu}.}
    \and
    Morteza Zadimoghaddam%
    \thanks{Google Research.
    Email: \href{mailto:zadim@google.com}{zadim@google.com}.}
}
\date{\today}

\newcommand{\primal}{\textsc{P}}
\newcommand{\dual}{\textsc{D}}

\newcommand{\type}{\tau}

\newcommand{\chosen}{\textrm{\em selected}}
\newcommand{\notchosen}{\textrm{\em not selected}}
\newcommand{\nonadapt}{\textrm{\em unknown}}

\newcommand{\sender}{\textrm{\em sender}}
\newcommand{\receiver}{\textrm{\em receiver}}

\newcommand{\kmin}{k_{\min}}
\newcommand{\kmax}{k_{\max}}

\newcommand{\exante}{\text{\em ex-ante}}
\newcommand{\expost}{\text{\em ex-post}}

\newcommand{\defeq}{\stackrel{\text{def}}{=}}

\newcommand{\E}{\mathbb{E}}

\newcommand{\argmax}{\arg\max}

\begin{document}

\begin{titlepage}
    \thispagestyle{empty}
    \maketitle
    \begin{abstract}
        \thispagestyle{empty}
        Online bipartite matching and its variants are among the most fundamental
problems in the online algorithms literature.  Karp, Vazirani, and
Vazirani~(STOC 1990) introduced an elegant algorithm for the unweighted problem
that achieves an optimal competitive ratio of $1-\sfrac{1}{e}$.  Later, Aggarwal
et al.~(SODA 2011) generalized their algorithm and analysis to the
vertex-weighted case.  Little is known, however, about the most general
edge-weighted problem aside from the trivial $\sfrac{1}{2}$-competitive greedy
algorithm.  In this paper, we present the first online algorithm that breaks
the long-standing $\sfrac{1}{2}$ barrier and achieves a competitive ratio of at least
$0.5086$.  In light of the hardness result of Kapralov, Post, and
Vondr\'ak~(SODA 2013) that restricts beating a~$\sfrac{1}{2}$ competitive ratio for the
more general problem of monotone submodular welfare maximization, our result
can be seen as strong evidence that edge-weighted bipartite matching is
strictly easier than submodular welfare maximization in the online setting.

The main ingredient in our online matching algorithm is a novel subroutine
called \emph{online correlated selection} (OCS), which takes a sequence of
pairs of vertices as input and selects one vertex from each pair.  Instead of
using a fresh random bit to choose a vertex from each pair, the OCS negatively
correlates decisions across different pairs and provides a quantitative measure
on the level of correlation.  We believe our OCS technique is of
independent interest and will find further applications in other online
optimization problems.

    \end{abstract}
\end{titlepage}

\newpage
\tableofcontents
\pagenumbering{gobble}
\clearpage
\pagenumbering{arabic}


\section{Introduction}
\label{sec:introduction}

Matchings are fundamental structures in graph theory that play
an indispensable role in combinatorial optimization.
For decades, there have been tremendous and ongoing efforts to
design more efficient algorithms for finding maximum matchings in terms of
their cardinality, and more generally, their total weight.
In particular, matchings in bipartite graphs have found countless applications
in settings where it is desirable to assign entities from one set to
those in another set (e.g., matching students to schools, physicians to hospitals,
computing tasks to servers, and impressions in online media to advertisers).
Due to the enormous growth of matching markets in digital domains,
efficient online matching algorithms have become increasingly important.
In particular, search engine companies
have created opportunities for online matching algorithms to
have a massive impact in multibillion-dollar advertising markets.
Motivated by these applications, we consider the problem of matching a set of
impressions that arrive one by one to a set of advertisers that are known in advance.
When an impression arrives, its edges to the advertisers are revealed
and an irrevocable decision has to be made about
to which advertiser the impression should be assigned.
Karp, Vazirani, and Vazirani~\cite{karp1990optimal} gave an elegant online algorithm called
\textsc{Ranking} to find matchings in unweighted bipartite graphs
with a \emph{competitive ratio} of $1-\sfrac{1}{e}$.
They also proved that this is the best achievable competitive ratio.
Further, Aggarwal et al.~\cite{aggarwal2011online} generalized their
algorithm to the vertex-weighted online bipartite matching problem
and showed that the $1-\sfrac{1}{e}$ competitive ratio is still attainable.

The edge-weighted case, however, is much more nebulous.
This is partly due to the fact that no competitive algorithm
exists without an additional assumption.
To see this, consider two instances of the edge-weighted problem, each with one advertiser and two impressions.
The edge weight of the first impression is $1$ in both instances,
and the weight of the second impression is $0$ in the first instance and $W$ in the second instance,
for some arbitrarily large $W$.
An online algorithm cannot distinguish between the two instances when the
first impression arrives,
but it has to decide whether or not to assign this impression to the advertiser.
Not assigning it gives a competitive ratio of $0$ in the first instance,
and assigning it gives an arbitrarily small competitive ratio of $\sfrac{1}{W}$ in the second.
This problem cannot be tackled unless assigning both impressions to the advertiser is an option.


In display advertising, assigning more impressions to an advertiser than they
paid for only makes them happier.
In other words, we can assign multiple impressions to any
given advertiser.
However, instead of achieving the weights of all the edges assigned to it,
we only acknowledge the maximum weight
(i.e., the objective equals the sum of the heaviest edge weight assigned to each advertiser).
This is equivalent to allowing the advertiser to dispose of previously matched
edges for free to make room for new, heavier edges.
This assumption is commonly known as the \emph{free disposal model}.
In the display advertising literature~\cite{feldman2009online,korula2013bicriteria},
the free-disposal assumption is well received and widely applied 
because of its natural economic interpretation.
Finally, edge-weighted online bipartite matching with free disposal
is a special case of the monotone submodular welfare maximization problem,
where we can apply known $\sfrac{1}{2}$-competitive greedy algorithms~\cite{fisher1978analysis,lehmann2006combinatorial}.

\subsection{Our Contributions}

Despite thirty years of research in online matching since the seminal work of
Karp et al.~\cite{karp1990optimal},
finding an algorithm for edge-weighted online bipartite matching that
achieves a competitive ratio greater than $\sfrac{1}{2}$
has remained a tantalizing open problem.
This paper gives a new online algorithm
and answers the question affirmatively,
breaking the long-standing $\sfrac{1}{2}$ barrier (under free disposal).

\begin{theorem}
There is a 0.5086-competitive algorithm for
edge-weighted online bipartite matching.
\end{theorem}

\noindent
Given the hardness result of Kapralov, Post, and Vondr\'ak~\cite{kapralov2013online}
that restricts beating a competitive ratio of $\sfrac{1}{2}$ for monotone submodular welfare maximization,
our algorithm shows that edge-weighted bipartite matching is strictly easier
than submodular welfare maximization in the online setting.

From now on,
we will use the more formal terminologies of offline and online vertices
in a bipartite graph instead of advertisers and impressions.
One of our main technical contributions is a novel algorithmic ingredient called
\emph{online correlated selection} (OCS), which is
an online subroutine that takes a sequence of pairs of vertices as input
and selects one vertex from each pair.
Instead of using a fresh random bit to make each of its decisions,
the OCS asks to what extent the decisions across different pairs can be negatively correlated,
and ultimately guarantees that a vertex appearing in~$k$ pairs is selected at least once
with probability strictly greater than $1 - 2^{-k}$.
See Section~\ref{sec:ocs-intro} for a short introduction
and Section~\ref{sec:ocs} for the full details.

Given an OCS,
we can achieve a better than $\sfrac{1}{2}$ competitive ratio for
unweighted online bipartite matching
with the following (barely) randomized algorithm.
For each online vertex, either pick a pair of offline neighbors
and let the OCS select one of them, or choose one offline neighbor deterministically.
More concretely, among the neighbors that have not been matched deterministically,
find the least-matched ones (i.e., those that have appeared in the least number of pairs).
Pick two if there are at least two of them;
otherwise, choose one deterministically.
We analyze this algorithm in Appendix~\ref{app:unweighted}.

Although the competitive ratio of the algorithm above is far worse than
the optimal $1-\sfrac{1}{e}$ ratio by Karp et al.~\cite{karp1990optimal},
it benefits from improved generalizability.
To extend this algorithm to the edge-weighted problem,
we need a reasonable notion of ``least-matched'' offline neighbors.
Suppose one neighbor's heaviest edge weight is either $1$ or $4$
each with probability $\sfrac{1}{2}$,
another neighbor's heaviest edge is $2$ with certainty,
and their edge weights with the current online vertex are both~$3$.
Which one is less matched?
To remedy this, we use the online primal-dual framework for matching problems
by Devanur, Jain, and Kleinberg~\cite{devanur2013randomized}, along with
an alternative formulation of the edge-weighted online bipartite matching problem
by Devanur et al.~\cite{DevanurHKMY/TEAC/2016}.
In short, we account for the contribution of each offline vertex
by weight-levels, and at each weight-level we consider the probability
that the heaviest edge matched to the vertex has weight at least this level.
This is the complementary cumulative distribution function (CCDF) of the heaviest edge weight,
and hence we call this the CCDF viewpoint.
Then for each offline neighbor,
we utilize the dual variables to compute an offer at each weight-level,
should the current online vertex be matched to it.
The neighbor with the largest net offer aggregating over all weight-levels is considered the ``least-matched''.
We introduce the online primal-dual framework and the CCDF
viewpoint in Section~\ref{sec:preliminary}.
Then we formally present our edge-weighted matching algorithm
in Section~\ref{sec:edge-weighted},
followed by its analysis.
Lastly, Appendix~\ref{app:hardness} includes hard instances that show
the competitive ratio of our algorithm is nearly tight.

\subsection{Related Works}

The literature of online weighted bipartite matching algorithms is extensive, but most of these works are devoted to achieving competitive ratios greater than $\sfrac{1}{2}$ by assuming that offline vertices have large capacities or that some stochastic information about the online vertices is known in advance.
Below we list the most relevant works and refer interested readers to the excellent survey of Mehta~\cite{mehta2013online}.
We note that there have recently been several significant advances in more general settings, including different arrival models and general (non-bipartite) graphs~\cite{huang2018match, gamlath2019beating,gamlath2019online, huang2019tight}. 

\paragraph{Large Capacities.}
%
The capacity of an offline vertex is the number of online vertices that can be assigned to it.
Exploiting the large-capacity assumption to beat $\sfrac{1}{2}$ dates back
two decades ago to Kalyanasundaram and Pruhs~\cite{kalyanasundaram2000optimal}. 
Feldman et al.~\cite{feldman2009online} gave a $(1-\sfrac{1}{e})$-competitive algorithm for Display Ads, which is equivalent to edge-weighted online bipartite matching assuming large capacities.
Under similar assumptions, the same competitive ratio was obtained for 
AdWords~\cite{mehta2005adwords, buchbinder2007online}, in which offline vertices have some budget constraint on the total weight that can be assigned to them rather than the number of impressions.
From a theoretical point of view, one of the primary goals in the online matching literature is to provide algorithms with competitive ratio greater than $\sfrac{1}{2}$ without making any assumption on the capacities of offline vertices.

\paragraph{Stochastic Arrivals.}
If we have knowledge about the arrival patterns of online vertices, we can often leverage this information to design better algorithms.
Typical stochastic assumptions include
assuming the online vertices are drawn from some known or unknown distribution~\cite{feldman2009online2,karande2011online,devanur2011near,haeupler2011online,
manshadi2012online,
mehta2012online,jaillet2013online},
or that they arrive in a random order~\cite{goel2008online,devanur2009adwords,feldman2010online, mahdian2011online,
mirrokni2012simultaneous,mehta2015online,huang2019online}.
These works achieve a $1-\varepsilon$ competitive ratio if the large capacity assumption holds in addition to the stochastic assumptions,
or at least $1-\sfrac{1}{e}$ for arbitrary capacities.
Korula, Mirrokni, and Zadimoghaddam~\cite{korula2018online} showed that the greedy algorithm is $0.505$-competitive for the more general problem of submodular welfare maximization if the online vertices arrive in a random order, without any assumption on the capacities.
The random order assumption is particularly justified because
Kapralov et al.~\cite{kapralov2013online} proved that beating $\sfrac{1}{2}$ for submodular welfare maximization in the oblivious adversary model implies $\textbf{NP} = \textbf{RP}$.





\section{Preliminaries}
\label{sec:preliminary}

The \emph{edge-weighted online matching} problem considers a bipartite graph $G
= (L, R, E)$, where $L$ and $R$ are the sets of vertices on the left-hand side
(LHS) and right-hand side (RHS), respectively, and $E \subseteq L \times R$ is
the set of edges.  Every edge $(i, j) \in E$ is associated with a nonnegative
weight $w_{ij} \ge 0$, and we can assume without loss of generality that this
is a complete bipartite graph, i.e., $E = L \times R$, by assigning zero
weights to the missing edges.

The vertices on the LHS are offline in that they are all known to the algorithm
in advance.  The vertices on the RHS, however, arrive online one at a time.
When an online vertex $j \in R$ arrives, its incident edges and their weights
are revealed to the algorithm, who must then irrevocably match $j$ to an
offline neighbor.  Each offline vertex can be matched any number of times, but
only the weight of its heaviest edge counts towards the objective.  This is
equivalent to allowing a matched offline vertex $i$, say, to $j$, to be
rematched to a new online vertex $j'$ with edge weight $w_{ij'} > w_{ij}$,
disposing of vertex $j$ and edge $(i, j)$ for free.
This assumption is known as the \emph{free disposal model}.

The goal is to maximize the total weight of the matching.  A randomized
algorithm is \emph{$\Gamma$-competitive} if its expected objective value
is at least $\Gamma$ times the offline optimal in hindsight,
for any instance of edge-weighted online matching.
We refer to $0 \le \Gamma \le 1$ as the
\emph{competitive ratio} of the algorithm.

\subsection{Complementary Cumulative Distribution Function Viewpoint}
\label{sec:histogram-view}

Next we describe an alternative formulation of the edge-weighted online matching
problem due to Devanur et al.~\cite{DevanurHKMY/TEAC/2016}
that captures the contribution of each offline vertex $i \in L$
to the objective in terms of the complementary cumulative
distribution function (CCDF) of the heaviest edge weight matched to $i$.
We refer to this approach as the \emph{CCDF viewpoint}.

For any offline vertex $i \in L$ and any weight-level $w \ge 0$, let $y_i(w)$ be
CCDF of the weight of the heaviest edge matched to $i$, i.e., the probability
that $i$ is matched to at least one online vertex~$j$ such that $w_{ij} \ge w$.
Then, $y_i(w)$ is a non-increasing function of $w$ that takes values between
$0$ and $1$.  Observe that $y_i(w)$ is a step function with polynomially many
pieces, because the number of pieces is at most the number of incident edges.
Hence, we will be able to maintain $y_i(w)$ in polynomial time.

The expected weight of the heaviest edge matched to $i$ then equals the area
under $y_i(w)$, i.e.:
\begin{equation}
    \label{eqn:histogram-expected-weight}
    \int_0^\infty y_i(w) dw~.
\end{equation}

\noindent
This follows from an alternative formula for the expected value
of a nonnegative random variable involving only its
cumulative distribution function.

We illustrative this idea with an example in Figure~\ref{fig:ccdf_viewpoint}.
Suppose an offline vertex $i$ has four online neighbors $j_1$, $j_2$, $j_3$,
and $j_4$ with edge weights $w_1 < w_2 < w_3 < w_4$.
Further, suppose that $j_1$ is matched to $i$ with certainty,
while $j_2$, $j_3$, and $j_4$ each have some probability of
being matched to $i$. (The latter events may be correlated.)
Next, suppose a new neighbor arrives whose edge weight is also $w_3$. 
The values of $y_i(w)$ are then increased for $w_1 < w \le w_3$ accordingly,
and the total area of the shaded regions is the increment in the
expected weight of the heaviest edge matched to vertex $i$.

\begin{figure}[t]
    \centering
    \begin{subfigure}{.4\textwidth}
    \includegraphics[width=0.85\textwidth]{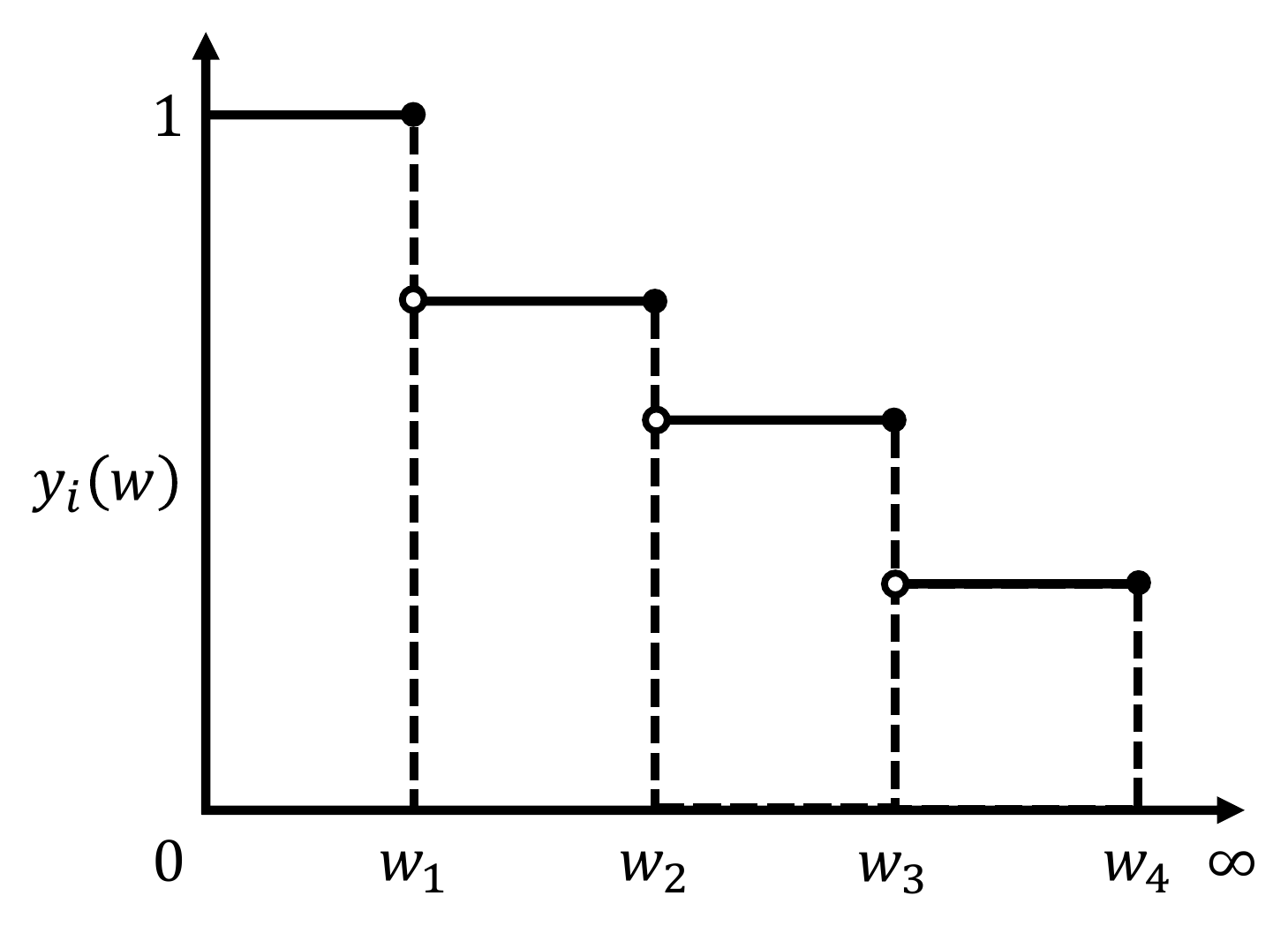}
    \label{fig:histogram}
    \end{subfigure}
    \begin{subfigure}{.4\textwidth}
    \includegraphics[width=0.85\textwidth]{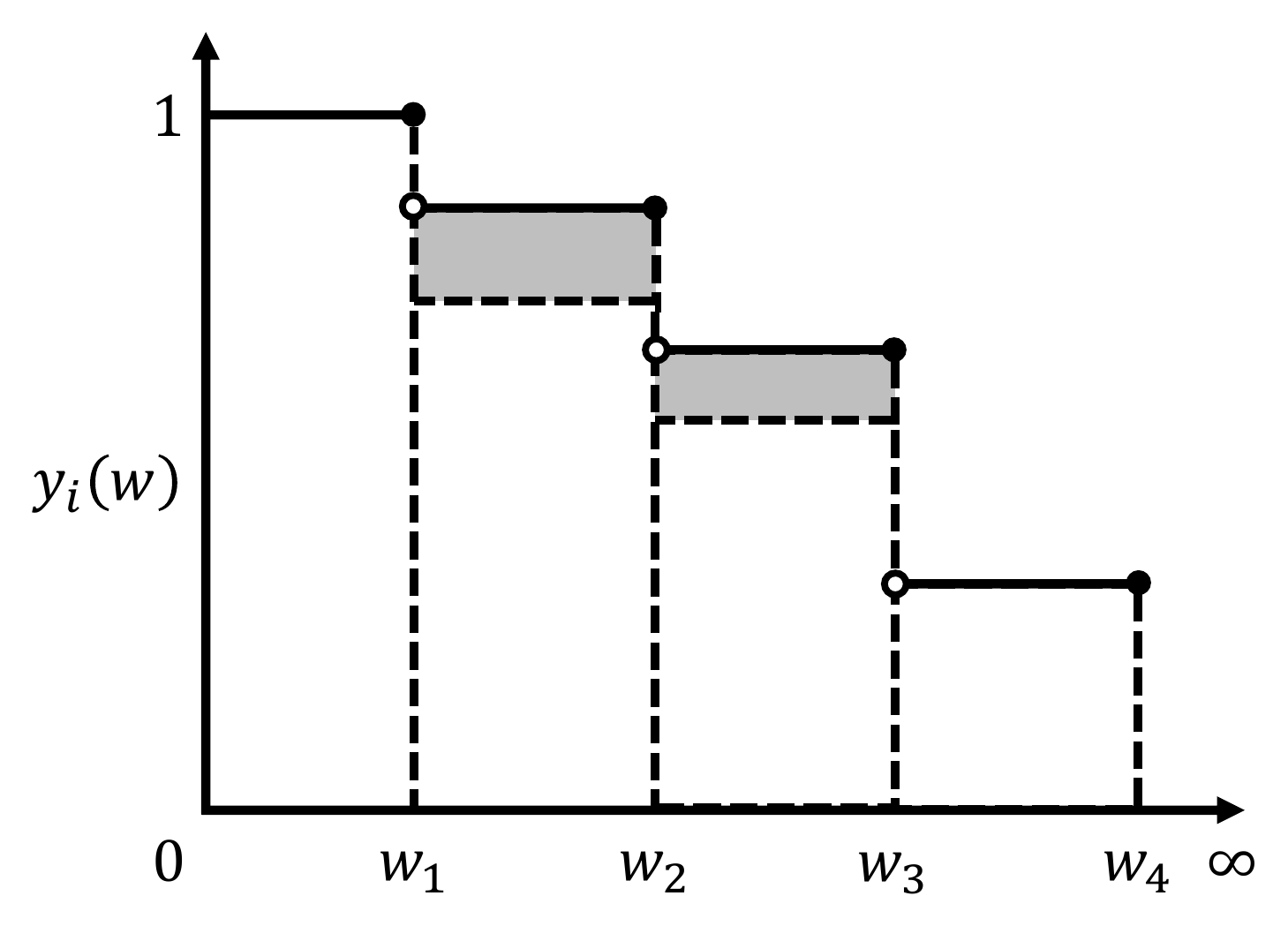}
    \label{fig:histogram-increment}
    \end{subfigure}
    \caption{
      Complementary cumulative distribution function (CCDF) viewpoint.
      The first function is the CCDF of vertex $i$,
      and the second function demonstrates how the CCDF of vertex $i$ is updated.
    }
    \label{fig:ccdf_viewpoint}
\end{figure}

\subsection{Online Primal-Dual Framework}

We analyze our algorithms using a linear program (LP) for
edge-weighted matching under the online primal-dual framework.  Consider the
standard matching LP and its dual below.
We interpret the primal variables $x_{ij}$
as the probability that $(i, j)$ is the heaviest edge matched to
vertex $i$.
\vspace{0.25cm}
\begin{minipage}{0.5\textwidth}
\begin{align*}
    \max \quad 
    & 
    \sum_{i \in L} \sum_{j \in R} w_{ij} x_{ij} \\
    \textrm{s.t.} \quad
    &
    \sum_{j \in R} x_{ij} \le 1 && \forall i \in L \\
    &
    \sum_{i \in L} x_{ij} \le 1 && \forall j \in R \\
    &
    x_{ij} \ge 0 && \forall i \in L, \forall j \in R
\end{align*}
\end{minipage}
\begin{minipage}{0.5\textwidth}
\begin{align*}
    \min \quad
    &
    \sum_{i \in L} \alpha_i + \sum_{j \in R} \beta_j \\[.5ex]
    \textrm{s.t.} \quad
    &
    \alpha_i + \beta_j \ge w_{ij} && \forall i \in L, \forall j \in R \\[2.5ex]
    &
    \alpha_i \ge 0 && \forall i \in L \\[2.5ex]
    &
    \beta_j \ge 0 && \forall j \in R
\end{align*}
\end{minipage}
\vspace{0.10cm}

\noindent
Let $\primal$ denote the primal objective.
If $x_{ij}$ is the probability that $(i, j)$ is the heaviest edge matched to
$i$, then $\primal$ also equals the objective of the algorithm.
Let $\dual$ denote the dual objective.

Online algorithms under the online primal-dual framework maintain not only a
matching but also a dual assignment (not necessarily feasible)
at all times subject to the conditions summarized below.

\begin{lemma}
    \label{lem:online-primal-dual}
    Suppose an online algorithm simultaneously maintains primal and dual
    assignments such that for some constant $0 \le \Gamma \le 1$, the following
    conditions hold at all times:
    \begin{enumerate}
        \item \textbf{\em Approximate dual feasibility:~}
          For any $i \in L$ and any $j \in R$,
          we have $\alpha_i + \beta_j \ge \Gamma \cdot w_{ij}$.
        \item \textbf{\em Reverse weak duality:~}
        The objectives of the primal and dual assignments 
        satisfy $\primal \ge \dual$.
    \end{enumerate}
    Then, the algorithm is $\Gamma$-competitive.
\end{lemma}

\begin{proof}
    By the first condition, the values $\Gamma^{-1} \alpha_i$ and
    $\Gamma^{-1} \beta_j$ form a feasible dual assignment whose objective
    equals $\Gamma^{-1} \dual$.
    By weak duality of linear programming, the objective of any feasible
    dual assignment upper bounds the optimal
    (i.e., $\dual$ is at least $\Gamma$ times
    the optimal).
    Applying the second condition now proves the lemma.
\end{proof}

\paragraph{Online Primal-Dual in the CCDF Viewpoint.}
In light of the CCDF viewpoint, for any offline vertex $i \in L$ and any
weight-level $w > 0$, we introduce and maintain new variables $\alpha_i(w)$
that satisfy:
\begin{equation}
    \label{eqn:alpha-by-weight} 
    \alpha_i = \int_0^\infty \alpha_i(w) dw~.
\end{equation}

\noindent
Accordingly, we rephrase approximate dual feasibility in
Lemma~\ref{lem:online-primal-dual} in the CCDF viewpoint as:
\begin{equation}
    \label{eqn:approximate-dual-feasible}
    \int_0^{\infty} \alpha_i(w) dw + \beta_j \ge \Gamma \cdot w_{ij}~.
\end{equation}

Concretely, at each step of our primal-dual algorithm,
$\alpha_i(w)$ is a piecewise constant function
with possible discontinuities at the weight-levels
$w \in \{w_{ij} \in E : \text{online vertex $j$ has arrived}\}$.
Initially, all of the $\alpha_i(w)$'s are the zero function.
Then, as each online vertex $j \in R$ arrives, if $j$ is potentially matched to
an offline candidate $i \in L$, the function values of $\alpha_{i}(w)$ are
systematically increased
according to the dual update rules in Section~\ref{sec:edge-weighted-algorithm}.
In contrast, each dual variable $\beta_j$ 
is a scalar value that is initialized to zero and increased only once
during the algorithm, at the time when $j$ arrives.

\section{Online Correlated Selection: An Introduction}
\label{sec:ocs-intro}

This section introduces our novel ingredient for online algorithms,
which we believe to be widely-applicable and of independent interest.
To motivate this technique, consider the following thought experiment 
in the case of \emph{unweighted} online matching, i.e., $w_{ij} \in \{0, 1\}$
for any $i \in L$ and any $j \in R$.

\paragraph{Deterministic Greedy.}
We first recall why all deterministic greedy algorithms that match each
online vertex to an unmatched offline neighbor are at most $\sfrac{1}{2}$-competitive.
Consider an instance with a graph that has two offline and two online vertices.
The first online vertex is adjacent to both offline vertices, and
the algorithm deterministically chooses one of them.
The second online vertex, however,
is only adjacent to the previously matched vertex.

\paragraph{Two-Choice Greedy with Independent Random Bits.}
We can avoid the problem above by matching the first online vertex randomly,
which improves the expected matching size from $1$ to $1.5$.
In this spirit, consider the following two-choice greedy algorithm.
When an online vertex arrives,
identify its neighbors that are least likely
to be matched (over the randomness in previous rounds).
If there is more than one such neighbor, choose any two,
e.g., lexicographically, and match to one with a fresh random bit.
Otherwise, match to the least-matched neighbor deterministically.
We refer to the former as a \emph{randomized round}
and the latter as a \emph{deterministic round}.
Since each randomized round uses a fresh random bit, this is equivalent to
matching to neighbors that have been chosen in the least number of randomized
rounds and in no deterministic round.
Unfortunately, this algorithm is also $\sfrac{1}{2}$-competitive
due to upper triangular graphs.  We defer this standard
example to Appendix~\ref{app:hardness}.

\paragraph{Two-choice Greedy with Perfect Negative Correlation.}
The last algorithm in this thought experiment is an imaginary variant of
two-choice greedy that perfectly and negatively correlates the randomized
rounds so that each offline vertex is matched with certainty after being a
candidate in two randomized rounds.
This is infeasible in general.
Nevertheless, if we assume feasibility then
this algorithm is $\sfrac{5}{9}$-competitive~\cite{huang2019unweighted}.
In fact, it is effectively the
$2$-matching algorithm of Kalyanasundaram and Pruhs~\cite{kalyanasundaram2000optimal}, by having two
copies of each online vertex and allowing offline vertices to be matched twice.

\begin{quote}
\emph{Can we use partial negative correlation to retain feasibility
and break the $\sfrac{1}{2}$ barrier?}
\end{quote}

We answer this question affirmatively by introducing an algorithmic ingredient
called \emph{online correlated selection} (OCS),
which allows us to quantify the negative correlation among randomized rounds.
Appendix~\ref{app:unweighted} provides an analysis of the two-choice greedy
algorithm powered by OCS in the unweighted case.
Furthermore,
Section~\ref{sec:edge-weighted} generalizes this approach to edge-weighted
online matching, achieving the first algorithm with a competitive
ratio that is provably greater than $\sfrac{1}{2}$.

\begin{definition}[$\gamma$-semi-OCS]
  Consider a set of ground elements.
  For any $\gamma \in [0, 1]$, a $\gamma$-semi-OCS is an online algorithm that takes as input a sequence of pairs of elements, and selects one per pair
  such that if an element appears in $k \ge 1$ pairs, it is selected at least once with probability at least:
  \[
      1 - 2^{-k} (1 - \gamma)^{k-1}
      ~.
  \]
\end{definition}

Using independent random bits is a $0$-semi-OCS,
and the perfect negative correlation in the
thought experiment corresponds to a $1$-semi-OCS,
although it is typically infeasible.
Our algorithms satisfy a stronger definition,
which considers any collection of pairs containing an element $i$.
This stronger definition is useful for generalizing to the edge-weighted
bipartite matching problem.

In the following definition, a subsequence (not necessarily contiguous)
of pairs containing element~$i$
is \emph{consecutive} if it includes all the pairs that contain element $i$
between the first and last pair in the subsequence.
Further, two subsequences
of pairs are \emph{disjoint} if no pair belongs to both of them.
For example, consider the sequence
$(\{a,i\}, \{b,i\}, \{c,d\}, \{e,i\}, \{i,z\})$.
The subsequences
$(\{a,i\},\{b,i\})$ and $(\{i,z\})$ are consecutive and disjoint,
but the subsequence $(\{a,i\}, \{b,i\}, \{i,z\})$
is not consecutive because it does not include the pair $\{e,i\}$.

\begin{definition}[$\gamma$-OCS]
  \label{def:ocs}
  Consider a set of ground elements.
  For any $\gamma \in [0, 1]$, a $\gamma$-OCS is an online algorithm that takes as input a sequence of pairs of elements, and selects one per pair such that for any element $i$ and any disjoint subsequences of $k_1, k_2, \dots, k_m$ consecutive pairs containing $i$, $i$ is selected in at least one of these pairs with probability at least:
  \[
      1 - \prod_{\ell=1}^m 2^{-k_\ell} (1 - \gamma)^{k_\ell-1}
      ~.
  \]
\end{definition}

\begin{theorem}
  \label{thm:ocs-main}
  There exists a $\frac{13\sqrt{13}-35}{108} > 0.1099$-OCS.
\end{theorem}

We defer the design and analysis of this OCS to Section~\ref{sec:ocs},
and instead describe a weaker $\sfrac{1}{16}$-OCS,
which is already sufficient for breaking the $\sfrac{1}{2}$ barrier
in edge-weighted online bipartite matching.

\begin{proof}[Proof Sketch of a $\sfrac{1}{16}$-OCS]
    Consider two sequences of random bits.
    The first set is used to construct a random matching among the pairs, where any two consecutive pairs (with respect to some element) are matched with probability $\sfrac{1}{16}$.
    Concretely, each pair is consecutive to at most four pairs, one before it and one after it for each of its two elements.
    For each pair, choose one of its consecutive pairs,
    each with probability $\sfrac{1}{4}$.
    Two consecutive pairs are matched if they choose each other.

    The second random sequence is used to select elements from the pairs.
    For any unmatched pair, choose one of its elements with a fresh random bit.
    For any two matched pairs, use a fresh random bit to choose an element
    in the first pair, and then make the opposite selection for the later one
    (i.e., select the common element if it is not selected in the earlier pair, and vice versa).
    Observe that even if two matched pairs are identical, there is no ambiguity in the opposite selection.

    Next, fix any element $i$ and any disjoint subsequences of $k_1, k_2, \dots, k_m$ consecutive pairs containing~$i$.
    We bound the probability that $i$ is never selected.
    If any two of the pairs are matched,
    $i$ is selected once in the two pairs.
    Otherwise, the selections from the pairs are independent,
    and the probability that $i$ is never selected
    is $\prod_{\ell=1}^m 2^{-k_\ell}$.
    Applying the law of total probability to the event that~$i$ is
    in a matched pair,
    it remains to upper bound the probability of having no such matched pairs by $\prod_{\ell=1}^m (1 - \sfrac{1}{16})^{k_\ell-1}$.
    Intuitively, this is because there are $k_\ell-1$ choices of two consecutive pairs within the $\ell$-th subsequence,
    each of which is matched with probability $\sfrac{1}{16}$.
    Further, these events are negatively dependent and therefore,
    the probability that none of them happens is upper bounded by the independent case.
    The formal analysis in Section~\ref{sec:ocs} substantiates this claim.
\end{proof}

\section{Edge-Weighted Online Matching}
\label{sec:edge-weighted}

This section presents an online primal-dual algorithm for the edge-weighted
online bipartite matching problem.
The algorithm uses a $\gamma$-OCS as a black box,
and its competitive ratio depends on the value of $\gamma$.
For $\gamma = \sfrac{1}{16}$ (as sketched in Section~\ref{sec:ocs-intro})
it is $0.505$-competitive, and for
$\gamma \approx 0.1099$ (as in Theorem~\ref{thm:ocs-main})
it is $0.5086$-competitive,
proving our main result about edge-weighted online matching.

\subsection{Online Primal-Dual Algorithm}
\label{sec:edge-weighted-algorithm}

The algorithm is similar to the two-choice greedy in the previous section.
It maintains an OCS with the offline vertices as the ground elements.
For each online vertex, the algorithm either
(1) matches it deterministically to one offline neighbor,
(2) chooses a pair of offline neighbors and matches to the
one selected by the OCS,
or (3) leaves it unmatched. 
We refer to the first case as
a \emph{deterministic} round, the second as a \emph{randomized} round, and the
third as an \emph{unmatched} round.

How does the algorithm decide whether it is a randomized, deterministic or unmatched round,
and how does it choose the candidate offline vertices?
%
We leverage the online primal-dual framework.
When an online vertex $j$ arrives,
it calculates for every offline vertex $i$ how much the dual
variable~$\beta_j$ would gain if $j$ is matched to $i$ in a deterministic round,
denoted as $\Delta_i^D \beta_j$, and similarly $\Delta_i^R \beta_j$ for a randomized round.
Then it finds $i^*$ with the maximum $\Delta_i^D \beta_j$, and $i_1, i_2$ with the maximum $\Delta_i^R \beta_j$.
If both $\Delta_{i_1}^R \beta_j + \Delta_{i_2}^R \beta_j$ and $\Delta_{i^*}^D \beta_j$ are negative, it leaves $j$ unmatched. 
If $\Delta_{i_1}^R \beta_j + \Delta_{i_2}^R \beta_j$ is nonnegative and greater than $\Delta_{i^*}^D \beta_j$,
it matches $j$ in a randomized round with $i_1$ and $i_2$ as the candidates using its OCS.
Finally, if $\Delta_{i^*}^D \beta_j$ is nonnegative and greater than $\Delta_{i_1}^R \beta_j + \Delta_{i_2}^R \beta_j$,
it matches $j$ to $i^*$ in a deterministic round.
See Algorithm~\ref{alg:primal-dual} for the formal definition of the algorithm.

It remains to explain how $\Delta_i^D \beta_j$ and $\Delta_i^R \beta_j$ are calculated.
For any offline vertex $i \in L$ and any weight-level $w > 0$,
let $k_i(w)$ be the number of randomized rounds in which $i$ has been
chosen and has edge weight at least $w$.
The values of $k_i(w)$ may change over time,
so we consider these values at the beginning of each online round.
The increments to the dual variables $\alpha_{i}(w)$ and $\beta_j$
depend on the values of $k_i(w)$
via the following \emph{gain-sharing} parameters,
which we determine later
using a factor-revealing LP to optimize the competitive ratio.
The gain-sharing values are listed at the end of this section in
Table~\ref{tab:lp-solution}.

\begin{itemize}
    \item $a(k)$ : Amortized increment in the dual variable $\alpha_i(w)$ 
      if $i$ is chosen as one of the two candidates in a randomized round
      in which its edge weight is at least $w$ and $k_i(w) = k$.
    \item $b(k)$ : Increment in the dual variable $\beta_j$ due to an offline
      vertex $i$ at weight-level $w \le w_{ij}$
      if~$j$ is matched in a randomized round with $i$ as one of the
      two candidates and $k_i(w) = k$.
\end{itemize}

\noindent
Note that these gain-sharing values $a(k)$ and $b(k)$
are instance independent
(i.e., they do not depend on the underlying graph)
and defined for all $k \in \mathbb{Z}_{\ge 0}$.
We interpret these parameters according to a gain-splitting rule.
If $i$ is one of the two candidates to be matched to $j$ in a randomized round,
the increase in the expected weight of the heaviest edge matched to $i$ equals
the integration of $y_i(w)$'s increments, for $0 < w \le w_{ij}$,
which can be related to the values of the $k_i(w)$'s.
We then lower bound the gain due to the increment of $y_i(w)$
using the definition of a $\gamma$-OCS
and split the gain into two parts,
$a(k_i(w))$ and $b(k_i(w))$.
The former is assigned to $\alpha_i(w)$ and the latter goes to $\beta_j$.

\begin{algorithm}[t]
    \caption{Online primal-dual edge-weighted bipartite matching algorithm.}
    \label{alg:primal-dual}
    \begin{algorithmic}
        \medskip
        \STATEx \textbf{State variables:}
        \begin{itemize}
            \item $k_i(w) \ge 0$ : The number of randomized rounds in which $i$ is a candidate
              and its edge weight is at least $w$;
            $k_i(w) = \infty$ if it has been chosen in a deterministic round in which its edge weight is at least $w$.
        \end{itemize}
        \smallskip
        \STATEx \textbf{On the arrival of an online vertex $j \in R$:}
        \begin{enumerate}
            \item For every offline vertex $i \in L$, compute $\Delta_i^R \beta_j$ and $\Delta_i^D \beta_j$ according to Eqn.~\eqref{eqn:beta-increment-randomized} and \eqref{eqn:beta-increment-deterministic}.
            \item Find $i_1, i_2$ with the maximum $\Delta_i^R \beta_j$.
            \item Find $i^*$ with the maximum $\Delta_i^D \beta_j$.
            \item If $0 > \Delta_{i_1}^R \beta_j + \Delta_{i_2}^R \beta_j$ and $\Delta_{i^*}^D \beta_j$, leave $j$ unmatched.
            \hspace*{\fill}
            \textbf{(unmatched)}
            \item If $\Delta_{i_1}^R \beta_j + \Delta_{i_2}^R \beta_j \ge \Delta_{i^*}^D \beta_j$ and $0$, let the OCS pick one of $i_1$ and $i_2$.
            \hspace*{\fill}
            \textbf{(randomized)}
        \item If $\Delta_{i^*}^D \beta_j > \Delta_{i_1}^R \beta_j$ and $\ge 0$, match $j$ to $i^*$.
            \hspace*{\fill}
            \textbf{(deterministic)}
            \item Update the $k_i(w)$'s accordingly.
        \end{enumerate}
    \end{algorithmic}
    \smallskip
\end{algorithm}

In fact, 
we prove at the end of this subsection
the following invariant about how the dual variables $\alpha_i(w)$ are incremented:
%
\begin{equation}
    \label{eqn:invariant-alpha}
    \alpha_i(w) \ge \sum_{0 \le \ell < k_i(w)} a(\ell)
    ~.
\end{equation}

Next, define $\Delta^R_i \beta_j$ to be:
\begin{equation}
    \label{eqn:beta-increment-randomized}
    \Delta_i^R \beta_j \defeq 
    \int_0^{w_{ij}} b\left(k_i(w)\right) dw - \frac{1}{2} \int_{w_{ij}}^{\infty} \sum_{0 \le \ell < k_i(w)} a(\ell) dw 
    ~.
\end{equation}

We should think of $\Delta_i^R \beta_j$ as the increase in the dual
variable $\beta_j$ due to offline vertex $i$,
if $i$ is chosen as one of the two candidates for $j$ in an randomized round.
The first term in Eqn.~\ref{eqn:beta-increment-randomized}
follows from the interpretation of $b(k)$ above (and would be
the only term in the unweighted case).
The second term is designed to cancel out the extra help we get from
the $\alpha_i(w)$'s at weight-levels $w > w_{ij}$
in order to satisfy approximate dual feasibility for the edge $(i,j)$.
Concretely, if $j$ is matched in a randomized round to two candidates at least
as good as $i$, our choice of $b(k)$'s ensures approximate dual feasibility
between $i$ and $j$ (i.e., the following inequality holds):
\[
    \int_0^{\infty} \alpha_i(w) dw + 2 \cdot \Delta_i^R \beta_j \ge \Gamma \cdot w_{ij}
    ~.
\]

%

Finally, for some $1 < \kappa < 2$, define the value of $\Delta_i^D \beta_j$ to be:
\begin{equation}
    \label{eqn:beta-increment-deterministic}
    \Delta_i^D \beta_j 
    \defeq \kappa \cdot \Delta_i^R \beta_j \\[1ex]
    = \kappa \int_0^{w_{ij}} b\left(k_i(w)\right) dw - \frac{\kappa}{2} \int_{w_{ij}}^{\infty} \sum_{0 \le \ell < k_i(w)} a(\ell) dw 
    ~.
\end{equation}

\noindent
For concreteness, readers can assume $\kappa = 1.5$.
The competitive ratio, however, is insensitive to the choice of $\kappa$
as long as it is neither too close to $1$ nor to $2$.
On the one hand, $\kappa > 1$ ensures that if the algorithm chooses a
randomized round with offline vertex $i_1$ and another vertex $i_2$ as the candidates,
the contribution from $i_2$ to $\beta_j$ must be at least a $\kappa
- 1$ fraction of what $i_1$ offers; otherwise, the algorithm would have preferred
a deterministic round with $i_1$ alone.
On the other hand, we have $\kappa < 2$ because otherwise a randomized round
would always be inferior to a deterministic round.
We further explain the definitions of
$\Delta_{i}^R \beta_j$ and $\Delta_{i}^D \beta_j$
in Subsection~\ref{sec:approximate-dual-feasible},
and we demonstrate how their terms interact 
when proving that the dual assignments always satisfy approximate dual
feasibility.

\paragraph{Primal Increments.}
We have defined the primal algorithm and, implicitly,
how the dual algorithm updates the $\beta_j$'s.
It remains to define the updates to $\alpha_i(w)$'s.
Before that, we first need to characterize the primal increment since the dual
updates are driven by it.
Recall that by the CCDF viewpoint:
\[
    \primal = \sum_{i \in L} \int_0^\infty y_i(w) dw
    ~.
\]

Since it is difficult to account for the exact CCDF $y_i(w)$ due to complicated correlations in the selections,
we instead consider a lower bound for it given by the $\gamma$-OCS.
A critical observation here is that the decisions made by the primal-dual
algorithm are deterministic, except for the randomness in the OCS.
In particular, its choices of $i_1$, $i_2$, $i^*$ and
the decisions about whether a round is unmatched, randomized, or deterministic
are independent of the selections in the OCS and
therefore
\emph{deterministic quantities governed solely by the
input graph and arrival order of the online vertices}.
Hence, we may view the sequence of pairs of candidates
as fixed.

For any offline vertex $i$ and any weight-level $w > 0$, consider the
randomized rounds in which $i$ is a candidate and has edge weight at least $w$.
Decompose these rounds into disjoint collections of, say, $k_1, k_2, \dots, k_m$ consecutive rounds.
By Definition~\ref{def:ocs}, vertex $i$ is selected by the $\gamma$-OCS
in at least one of these rounds with probability at least:
\begin{equation}
\label{eqn:y-bar-def}
    \bar{y}_i(w) \defeq 1 - \prod_{\ell=1}^m 2^{-k_\ell} \left(1 - \gamma\right)^{k_\ell-1}
    ~.
\end{equation}

\noindent
Accordingly, we will use the following surrogate primal objective:
\[
    \bar{\primal} = \sum_{i \in L} \int_0^\infty \bar{y}_i(w) dw
    ~.
\]

\begin{lemma}
    \label{lem:surrogate-primal}
    The primal objective is lower bounded by the surrogate, i.e., $\bar{\primal} \le \primal$.
\end{lemma}

\noindent
It will often be more convenient to consider the following characterization of $\bar{y}_i(w)$:
\begin{itemize}
    \item Initially, let $\bar{y}_i(w) = 0$.
    \item If $i$ is matched in a deterministic round in which its edge weight is at least $w$, let $\bar{y}_i(w) = 1$.
    \item If $i$ is chosen in a randomized round
      in which its edge weight is at least $w$,
      further consider~$w'$, its edge weight in the previous round involving $i$;
      let $w' = 0$ if it is the first randomized round involving~$i$.
    Then, decrease the gap $1 - \bar{y}_i(w)$
    by a $\sfrac{1}{2} ( 1 - \gamma )$ factor if $w' \ge w$, i.e., if it is the second or later pair of a collection of consecutive pairs containing $i$ with edge weight at least~$w$; 
    otherwise, decrease $1 - \bar{y}_i(w)$ by $\sfrac{1}{2}$,
    to account for the $-1$ in the exponent of $1 - \gamma$
    in Eqn~\ref{eqn:y-bar-def}.
\end{itemize}

\begin{lemma}
    \label{lem:unmatched-portion-lower-bound}
    For any offline vertex $i$ and any weight-level $w > 0$, we have:
    \[
        1 - \bar{y}_i(w) \ge 2^{-k_i(w)} \left(1 - \gamma\right)^{\max\{k_i(w)-1, 0\}}
        ~.
    \]
\end{lemma}

\begin{proof}
    Initially, $1 - \bar{y}_i(w)$ equals $1$.
    Then, it decreases by $\sfrac{1}{2}$ in the first randomized round involving $i$ with edge weight at least $w$, and by at most $\sfrac{1}{2} \left( 1 - \gamma \right)$ in each of the subsequent $k_i(w) - 1$ rounds.
\end{proof}

This is equivalent to a lower bound of the increment in $y_i(w)$ in a deterministic round.

\begin{lemma}
    \label{lem:primal-increment-deterministic}
    For any offline vertex $i$ and any weight-level $w > 0$, if $i$ is matched in a deterministic round in which its edge weight is at least $w$, the increment in $\bar{y}_i(w)$ is at least:
    \[
        2^{-k_i(w)} \left(1 - \gamma\right)^{\max\{k_i(w)-1,0\}}
        ~.
    \]
\end{lemma}

\begin{lemma}
    \label{lem:primal-increment-randomized}
    For any offline vertex $i$ and any weight-level $w > 0$, if $i$ is chosen as a candidate in a randomized round in which its edge weight is at least $w$, the increment in $\bar{y}_i(w)$ is at least:
    \[ 
        2^{-k_i(w)-1} \left(1 - \gamma\right)^{\max\{k_i(w)-1, 0\}} 
        ~.
    \]
    Suppose further that vertex $i$'s edge weight is also at least $w$
    in the last randomized round involving $i$.
    Then, it follows that $k_i(w) \ge 1$ and the increment in $\bar{y}_i(w)$
    is at least:
    \[
        2^{-k_i(w)-1} \left(1 - \gamma\right)^{k_i(w)-1} \left(1 + \gamma\right)
        ~.
    \]
\end{lemma}

\begin{proof}
    By definition, $1 - \bar{y}_i(w)$ decreases by a factor of either $\sfrac{1}{2} (1 - \gamma)$ or $\sfrac{1}{2}$ in a randomized round, depending on whether vertex $i$'s edge weight is at least $w$ the last time it is chosen in a randomized round.
    Therefore, the increment in $\bar{y}_i(w)$ is either a $\sfrac{1}{2} (1 + \gamma)$ fraction of $1 - \bar{y}_i(w)$, or a $\sfrac{1}{2}$ fraction. 
    Putting this together with the lower bound for
    $1 - \bar{y}_i(w)$ in Lemma~\ref{lem:unmatched-portion-lower-bound} proves the lemma.
\end{proof}

\paragraph{Dual Updates to Online Vertices.}
Consider any online vertex $j \in R$ at the time of its arrival.
The dual variable $\beta_j$ will only increase at the end of this round,
depending on the type of assignment.
If $j$ is left unmatched, then the value of~$\beta_j$ remains zero.
If $j$ is matched in a randomized round,
set $\beta_j = \Delta_{i_1}^R \beta_j + \Delta_{i_2}^R \beta_j$.
Lastly, if $j$ is matched in a deterministic round, set
$\beta_j = \Delta_{i^*}^D \beta_j$.

\paragraph{Dual Updates to Offline Vertices: Proof of Eqn.~\eqref{eqn:invariant-alpha}.}
%
Fix any offline vertex $i \in L$.
Suppose that $i$ is matched in a \emph{deterministic round} in which its edge weight is $w_{ij}$.
Then, for any weight-level $w > w_{ij}$, the value of $k_i(w)$ stays the same,
so we leave $\alpha_i(w)$ unchanged.
On the other hand,
for any weight-level $w \le w_{ij}$,
the value of $k_i(w)$ becomes $\infty$ by definition.
Therefore, to maintain the invariant in Eqn.~\eqref{eqn:invariant-alpha},
we increase $\alpha_i(w)$ for each weight-level $w \le w_{ij}$ by:
\begin{equation}
    \label{eqn:alpha-increment-deterministic}
    \sum_{\ell = k_i(w)}^{\infty} a(\ell)
    ~.
\end{equation}

The updates in \emph{randomized rounds} are more subtle.
Suppose $i$ is one of the two candidates in a randomized round in which its edge weight is $w_{ij}$.
Further consider $i$'s edge weight the last time it was chosen in a randomized round, denoted as $w'$; 
let $w' = 0$ if this is the first randomized round involving vertex $i$.
Then, $w_{ij}$ and $w'$ partition the weight-levels $w > 0$ into up to three subsets,
each of which requires a different update rule for $\alpha_i(w)$.
Concretely, the algorithm increase $\alpha_i(w)$ by:
\begin{equation}
    \label{eqn:alpha-increment-randomized}
    \begin{cases}
        a\left(k_i(w)\right) 
        &
        \text{if } 0 < w \le w_{ij}, w' \textrm{ or } k_i(w) = 0 ~; \\
        a\left(k_i(w)\right) - 2^{-k_i(w)-1} \left(1 - \gamma\right)^{k_i(w)-1} \gamma 
        &
        \text{if } w' < w \le w_{ij} \textrm{ and } k_i(w) \ge 1 ~; \\
        2^{-k_i(w)-1} \left(1 - \gamma\right)^{k_i(w)-1} \gamma
        &
        \text{if } w > w_{ij} \textrm{ and } k_i(w) \ge 1 ~.
    \end{cases}
\end{equation}

The first case is straightforward---we simply increase $\alpha_i(w)$ by $a\left(k_i(w)\right)$
to maintain the invariant in Eqn.~\eqref{eqn:invariant-alpha}.
Observe that this is the only case in the unweighted problem.

For a weight-level $w$ that falls into the second case (if there is any),
the increment in $\alpha_i(w)$ is smaller than the first case
by $2^{-k_i(w)-1} (1 - \gamma)^{k_i(w)-1} \gamma$.
This is the difference between the lower bounds for the increments
in $\bar{y}_i(w)$ in Lemma~\ref{lem:primal-increment-randomized}, depending on
whether $i$'s edge weight was at least~$w$ the last time it
was chosen in a randomized round.
Since the increase in the surrogate primal objective $\bar{\primal}$ due to
vertex $i$ and weight-level $w$ (when $w' < w$)
is less than the first case of Eqn.~\eqref{eqn:alpha-increment-randomized},
we subtract this difference
from the increment in $\alpha_i(w)$
so that the update to $\beta_j$ is unaffected.

How can we still maintain the invariant in Eqn.~\eqref{eqn:invariant-alpha}
given the subtraction in the second case?
Observe that if the second case happens, the same weight-level must fall into the third case in the previous randomized round in which $i$ is involved. 
Thus, an equal amount is prepaid to each $\alpha_i(w)$ in the previous round.
This give-and-take in the offline dual vertex updates
becomes clear when we prove reverse weak duality in the next subsection.

\subsection{Online Primal-Dual Analysis: Reverse Weak Duality}
\label{sec:reverse-weak-duality}

This subsection derives a set of sufficient conditions under which the
increment in the surrogate primal $\bar{\primal}$ is at least that of the
dual solution $\dual$.
Reverse weak duality then follows from $\primal \ge \bar{\primal} \ge \dual$.

\paragraph{Deterministic Rounds.}
Suppose $j$ is matched to $i$ in a deterministic round.
Using the lower bound for the increase of $\bar{\primal}$
in Lemma~\ref{lem:primal-increment-deterministic},
the increase of the $\alpha_i(w)$'s in Eqn.~\eqref{eqn:alpha-increment-deterministic},
and a lower bound for $\beta_j$ by dropping the second term in Eqn.~\eqref{eqn:beta-increment-deterministic}, we need:
\[
    \int_0^{w_{ij}} \sum_{\ell = k_i(w)}^\infty a(\ell) dw + \kappa \int_0^{w_{ij}} b\left(k_i(w)\right) dw
    \le 
    \int_0^{w_{ij}} 2^{-k_i(w)} \left(1-\gamma\right)^{\max\{k_i(w)-1, 0\}} dw
    ~.
\]

\noindent
We will ensure the inequality locally at every weight-level,
so it suffices to have:
\begin{equation}
    \label{eqn:gain-split-deterministic}
    \forall k \ge 0 \quad : \quad 
    \sum_{\ell = k}^\infty a(\ell) + \kappa \cdot b(k)
    \le
    2^{-k} \left(1-\gamma\right)^{\max\{k-1, 0\}}
    ~.
\end{equation}

\paragraph{Randomized Rounds.}
Now suppose $j$ is matched 
with candidates $i_1, i_2$ in a randomized round.
We show that the increment in $\bar{\primal}$ due to $i_1$ is at least
the increase in the $\alpha_{i_1}(w)$'s plus its contribution to $\beta_j$
(i.e., $\Delta_{i_1}^R \beta_j$).
This also holds for $i_2$ by symmetry, and
together they prove reverse weak duality.

Let $w_1$ be the edge weight of $i \gets i_1$ in this round,
and let $w_1'$ be its edge weight the last time it was chosen in a randomized round;
set $w_1' = 0$ if this has not happened.
Then, $w_1$ and $w_1'$ partition the weight-levels $w > 0$
into three subsets corresponding to the three cases
for incrementing the dual variables $\alpha_i(w)$ in a randomized round,
as in Eqn.~\eqref{eqn:alpha-increment-randomized}

The \emph{first case} is when $w \le w_1, w_1'$ or $k_{i}(w) = 0$.
By Lemma~\ref{lem:primal-increment-randomized}, the increase in $\bar{\primal}$
due to vertex~$i$ at weight-level $w$ is at least:
\[
    \begin{cases}
      \frac{1}{2} & \text{if } k_{i}(w) = 0 ~; \\
        2^{-k_{i}(w)-1} \left(1 - \gamma\right)^{k_{i}(w)-1} \left(1 + \gamma\right)
          & \text{if } k_{i}(w) \ge 1 \textrm{ and } w \le w_1, w_1' ~.
    \end{cases}
\]

\noindent
By the first case of Eqn.~\eqref{eqn:alpha-increment-randomized},
the increase in $\alpha_{i}(w)$ is $a(k_{i}(w))$.
Finally, the contribution to the first term of
$\beta_j = \Delta_{i}^R \beta_j + \Delta_{i_2}^R \beta_j$,
at weight-level $w$, in Eqn.~\eqref{eqn:beta-increment-randomized} is $b(k_{i}(w))$.
Hence, it suffices to ensure:
\begin{equation}
    \label{eqn:gain-split-randomized}
    a(0) + b(0) \le \frac{1}{2}
    \quad\text{and}\quad
    \forall k \ge 1: a(k) + b(k) \le 2^{-k-1} \left(1 - \gamma\right)^{k-1} \left(1 + \gamma\right)
    ~.
\end{equation}

The \emph{second case} is when $w_1' < w \le w_1$ and $k_{i}(w) \ge 1$.
By Lemma~\ref{lem:primal-increment-randomized}, the increment in $\bar{\primal}$
due to~$i$ at weight-level $w$ is at least $2^{-k_{i}(w)-1} (1 - \gamma)^{k_{i}(w)-1}$.
By the second case of Eqn.~\eqref{eqn:alpha-increment-randomized},
the increase in $\alpha_{i}(w)$ is $a(k_{i}(w)) - 2^{-k_{i}(w)-1} (1 - \gamma)^{k_{i}(w)-1} \gamma$.
Finally, the contribution to the first term of $\beta_j$, at weight-level $w$, is $b(k_{i}(w))$.
Hence, we need:
\[
  a\left(k_{i}(w)\right) - 2^{-k_{i}(w)-1} \left(1 - \gamma\right)^{k_{i}(w)-1} \gamma + b\left(k_{i}(w)\right) 
    \le 
    2^{-k_{i}(w)-1} \left(1 - \gamma\right)^{k_{i}(w)-1} 
    ~.
\]

\noindent
Rearranging the second term to the RHS gives us the same conditions
as the second part of Eqn.~\eqref{eqn:gain-split-randomized}.

The \emph{third case} is when $w > w_1$ and $k_{i}(w) \ge 1$.
The increment in $\bar{\primal}$ due to $i$ at weight-level $w$ is~0.
By the last case of Eqn.~\eqref{eqn:alpha-increment-randomized},
the increase in $\alpha_{i}(w)$ is $2^{-k_{i}(w)-1} (1 - \gamma)^{k_{i}(w)-1} \gamma$.
The negative contribution from the second term of
$\beta_j$, at weight-level $w$, is $\frac{1}{2} \sum_{0 \le \ell < k_{i}(w)} a(\ell)$.
Hence, we need:
\[
  2^{-k_{i}(w)-1} \left(1 - \gamma\right)^{k_{i}(w)-1} \gamma - \frac{1}{2} \sum_{0 \le \ell < k_{i}(w)} a(\ell) \le 0
    ~.
\]
The first term is decreasing in $k_{i}(w)$ and the second is increasing,
so it suffices to consider $k_{i}(w) = 1$:
\begin{equation}
    \label{eqn:gain-split-prepaid}
    a(0) \ge \frac{\gamma}{2}
    ~.
\end{equation}

\subsection{Online Primal-Dual Analysis: Approximate Dual Feasibility}
\label{sec:approximate-dual-feasible}

This subsection derives a set of conditions that are sufficient
for approximate dual feasibility, i.e., Eqn.~\eqref{eqn:approximate-dual-feasible}.
Start by fixing any $i \in L$ and any $j \in R$, and also the
values of the $k_i(w)$'s when $j$ arrives.

\paragraph{Boundary Condition at the Limit.}
%
First, it may be the case that $k_i(w) = \infty$ for all $0 < w \le w_{ij}$ and $j$ is unmatched.
This means $\beta_j = 0$
in this round and thus, the contribution from the $\alpha_i(w)$'s alone
must ensure approximate dual feasibility.
To do so, we will ensure that the value of $\alpha_i(w)$ is at least $\Gamma$ whenever $k_i(w) = \infty$.
By the invariant in Eqn.~\eqref{eqn:invariant-alpha}, it suffices to have:
\begin{equation}
    \label{eqn:feasibility-alpha-infty}
    \sum_{\ell = 0}^\infty a(\ell) \ge \Gamma
    ~.
\end{equation}

Next, we consider five different cases that depend on whether the round of $j$ 
is randomized, deterministic or unmatched, and if $i$ is chosen as a candidate.
We first analyze the cases when $j$ is in a randomized round,
and then we will show that the other cases only require weaker conditions.

\paragraph{Case 1: Round of $j$ is a randomized, $i$ is not chosen.}
%
By definition, $\beta_j = \Delta_{i_1}^R \beta_j + \Delta_{i_2}^R \beta_j$.
Since $i$ is not chosen, both terms on the RHS are at least $\Delta_i^R \beta_j$.
Using the definition of $\Delta_i^R \beta_j$ in Eqn.~\eqref{eqn:beta-increment-randomized}
and lower bounding $\alpha_i(w)$ by Eqn.~\eqref{eqn:invariant-alpha},
approximate dual feasibility in Eqn.~\eqref{eqn:approximate-dual-feasible} reduces to:
\[
    \int_0^{w_{ij}} \sum_{0 \le \ell < k_i(w)} a(\ell) dw + 2 \int_0^{w_{ij}} b\left(k_i(w)\right) dw \ge \Gamma \cdot w_{ij}
    ~.
\]

\noindent
We will again ensure this inequality at every weight-level.
Therefore, it suffices to have:
\begin{equation}
    \label{eqn:feasibility-randomized-not-chosen}
    \forall k \ge 0 \quad : \qquad \sum_{0 \le \ell < k} a(\ell) + 2 \cdot b(k) \ge \Gamma
    ~.
\end{equation}

\paragraph{Case 2: Round of $j$ is randomized, $i$ is chosen.}
By symmetry, suppose WLOG that $i \gets i_2$ and $i_1$ is the other candidate.
By definition, $\beta_j = \Delta_{i}^R \beta_j + \Delta_{i_1}^R \beta_j$.
Next, we derive a lower bound only in terms of $\Delta_i^R \beta_j$.
Since the algorithm does not choose a deterministic round with $i$ alone,
we have $\Delta_{i}^R \beta_j + \Delta_{i_1}^R \beta_j \ge \Delta_i^D \beta_j$.
Further, we have $\Delta_i^D \beta_j = \kappa \cdot \Delta_i^R \beta_j$ by Eqn.~\eqref{eqn:beta-increment-deterministic}. 
Combining these, we have $\beta_j \ge \kappa \cdot \Delta_i^R \beta_j$.
Finally, by the definition of $\Delta_i^R \beta_j$ in Eqn.~\eqref{eqn:beta-increment-randomized}, $\beta_j$ is at least:
\[
    \kappa \cdot \bigg( \int_0^{w_{ij}} b\big(k_i(w)\big) dw - \frac{1}{2} \int_{w_{ij}}^\infty \sum_{0 \le \ell < k_i(w)} a(\ell) dw \bigg) 
    ~.
\]

Lower bounding the $\alpha_i(w)$'s is more subtle.
Recall that $k_i(w)$ denotes the value at the beginning of the round when $j$ arrives.
Thus, the value of $k_i(w)$ increases by $1$ for any weight-level
$0 < w \le w_{ij}$ and stays the same for any other weight-level $w > w_{ij}$.
Therefore, the contribution of the $\alpha_i(w)$'s
to approximate dual feasibility is at least:
\[
    \int_0^{w_{ij}} \sum_{0 \le \ell \le k_i(w)} a(\ell) dw + \int_{w_{ij}}^\infty \sum_{0 \le \ell < k_i(w)} a(\ell) dw
    ~.
\]

Finally, since $\kappa < 2$,
the net contribution from weight-levels $w > w_{ij}$ is nonnegative, so we can drop them.
Then approximate dual feasibility as in Eqn.~\eqref{eqn:approximate-dual-feasible} becomes:
\[
    \int_0^{w_{ij}} \left( \sum_{0 \le \ell \le k_i(w)} a(\ell) + \kappa \cdot b\left(k_i(w)\right) \right) dw
    \ge
    \Gamma \cdot w_{ij}
    ~.
\]

Thus, it suffices to ensure the inequality locally at every weight-level:
\begin{equation}
    \label{eqn:feasibility-randomized-chosen}
    \forall k \ge 0 \quad : \qquad \sum_{0 \le \ell \le k} a(\ell) + \kappa \cdot b(k) \ge \Gamma
    ~.
\end{equation}

\noindent
There are two differences between Eqn.~\eqref{eqn:feasibility-randomized-not-chosen}
and Eqn.~\eqref{eqn:feasibility-randomized-chosen}
First, the summation above includes $\ell = k$.
We can do this because $i$ is one of the two candidates and therefore, $k_i(w)$ increases by $1$ in the round of $j$ for any weight-level $w \le w_{ij}$.
Second, the $\kappa$ coefficient for the second term is smaller.

\paragraph{Case 3: Round of $j$ is deterministic, $i$ is not chosen.}
By definition, $\beta_j = \Delta_{i^*}^D \beta_j$.
Next, we derive a lower bound in terms of $\Delta_i^R \beta_j$.
Since the algorithm does not choose a randomized round with $i$ and $i^*$
as the two candidates,
we have $\Delta_{i^*}^D \beta_j > \Delta_{i^*}^R \beta_j + \Delta_i^R \beta_j$.
By Eqn.~\eqref{eqn:beta-increment-deterministic} and $\kappa < 2$,
we have $\Delta_{i^*}^R \beta_j > \frac{1}{2} \cdot \Delta_{i^*}^D \beta_j$.
Here, we use the fact that $\Delta_{i^*}^D \beta_j \ge 0$, because $i^*$ is chosen in a deterministic round.
Putting this together gives us
$\beta_j = \Delta_{i^*}^D \beta_j > 2 \cdot \Delta_i^R \beta_j$, which is
identical to the lower bound in the first case.
Therefore, approximate dual feasibility is guaranteed by Eqn.~\eqref{eqn:feasibility-randomized-not-chosen}.

\paragraph{Case 4: Round of $j$ is deterministic, $i$ is chosen.}
For any $0 < w \le w_{ij}$, we have $k_i(w) = \infty$ after this round.
Therefore, approximate dual feasibility follows from the contribution
of the $\alpha_i(w)$'s alone due to the invariant in Eqn.~\eqref{eqn:invariant-alpha}
and the boundary condition in Eqn.~\eqref{eqn:feasibility-alpha-infty}.

\paragraph{Case 5: Round of $j$ is unmatched.}
By definition, $\beta_j = 0$.
Moreover, $\Delta_i^D \beta_j < 0$ because the algorithm chooses to leave $j$ unmatched, which further implies $\Delta_i^R \beta_j < 0$ by Eqn.~\eqref{eqn:beta-increment-deterministic}.
Therefore, we have $\beta_j \ge 2 \cdot \Delta_i^R \beta_j$,
identical to the lower bound in the first case.
Thus, approximate dual feasibility
is guaranteed by Eqn.~\eqref{eqn:feasibility-randomized-not-chosen}.

\subsection{Optimizing the Gain-Sharing Parameters}
\label{sec:gain-sharing}

To optimize the competitive ratio $\Gamma$ in the above online primal-dual analysis,
it remains to solve for the gain sharing parameters $a(k)$ and $b(k)$ using
the following LP:
\begin{align*}
    \textrm{maximize} \quad & 
    \Gamma \\
    \textrm{subject to} \quad & \text{Eqn.~\eqref{eqn:gain-split-deterministic}, \eqref{eqn:gain-split-randomized}, \eqref{eqn:gain-split-prepaid}, \eqref{eqn:feasibility-alpha-infty}, \eqref{eqn:feasibility-randomized-not-chosen}, and \eqref{eqn:feasibility-randomized-chosen}}
\end{align*}

\noindent
We obtain a lower bound on the competitive ratio by solving
a more restricted LP, which is finite.
In particular,
we set $a(k) = b(k) = 0$ for all $k > \kmax$ for some sufficiently large integer $\kmax$,
so that it becomes:
\begin{align*}
    \textrm{maximize} \quad & 
    \Gamma \\[1ex]
    \textrm{subject to} \quad & 
    \sum_{k \le \ell \le \kmax} a(\ell) + \kappa \cdot b(k) \le 2^{-k} \left(1 - \gamma\right)^{\max\{k-1,0\}}  && \forall 0 \le k \le \kmax \\
    &
    a(0) + b(0) \le \frac{1}{2} \notag \\[1ex]
    & 
    a(k) + b(k) \le 2^{-k-1} \left(1 - \gamma\right)^{k-1} \left(1 + \gamma\right) && \forall 1 \le k \le \kmax \\[2ex]
    &
    a(0) \ge \frac{\gamma}{2} \\[1ex]
    &
    \sum_{0 \le \ell \le \kmax} a(\ell) \ge \Gamma \\
    &
    \sum_{0 \le \ell < k} a(\ell) + 2 \cdot b(k) \ge \Gamma && \forall 0 \le k \le \kmax \\
    &
    \sum_{0 \le \ell \le k} a(\ell) + \kappa \cdot b(k) \ge \Gamma && \forall 0 \le k \le \kmax \\
    & 
    a(k), b(k) \ge 0 && \forall 0 \le k \le \kmax
\end{align*}

We present an approximately optimal solution to the finite LP in
Table~\ref{tab:lp-solution-weak} 
with $\gamma = \sfrac{1}{16}$, $\kappa = \sfrac{3}{2}$,
and $\kmax = 8$, which gives $\Gamma > 0.505$.
We also tried different values of $\kappa = 1 + \sfrac{\ell}{16}$,
for $0 \le \ell \le 16$.
If $\kappa = 1$ or $\kappa = 2$, then $\Gamma = 0.5$;
if $\kappa = 1 + \sfrac{15}{16}$, then $\Gamma \approx 0.5026$;
for all other values of $\kappa$, $\Gamma > 0.505$.
Hence, the analysis is robust to the choice of
$\kappa$, so long as it is neither too close to $1$ nor to $2$.
In Table~\ref{tab:lp-solution-main} we present
an approximately optimal solution under the same setting except we use a larger
$\gamma = \frac{13\sqrt{13}-35}{108} > 0.1099$ as in Theorem~\ref{thm:ocs-main},
which leads to an improved competitive ratio $\Gamma > 0.5086$.%
\footnote{All of our source code is available at
\href{https://github.com/fahrbach/edge-weighted-online-bipartite-matching}{https://github.com/fahrbach/edge-weighted-online-bipartite-matching}.}

\begin{table}[t]
    \centering
    \renewcommand{\arraystretch}{1.2}
    \begin{subtable}{.4\textwidth}
        \centering
        \begin{tabular}{c|cc}
            \hline
            $k$ & $a(k)$ & $b(k)$ \\
            \hline
            $0$ & $0.24748256$ & $0.25251744$ \\
            $1$ & $0.13684883$ & $0.12877617$ \\
            $2$ & $0.06415997$ & $0.06035174$ \\
            $3$ & $0.03009310$ & $0.02827176$ \\
            $4$ & $0.01413332$ & $0.01322521$ \\
            $5$ & $0.00666576$ & $0.00615855$ \\
            $6$ & $0.00318572$ & $0.00282566$ \\
            $7$ & $0.00158503$ & $0.00123280$ \\
            $8$ & $0.00088057$ & $0.00044028$ \\
            \hline
        \end{tabular}
        \caption{$\gamma = \sfrac{1}{16}$, $\Gamma = 0.50503484$}
        \label{tab:lp-solution-weak}
    \end{subtable}
    \begin{subtable}{.4\textwidth}
        \centering
        \begin{tabular}{c|cc}
            \hline
            $k$ & $a(k)$ & $b(k)$ \\
            \hline
            $0$ & $0.24566361$ & $0.25433639$ \\
            $1$ & $0.14597716$ & $0.13150459$ \\
            $2$ & $0.06497349$ & $0.05851601$ \\
            $3$ & $0.02892807$ & $0.02602926$ \\
            $4$ & $0.01289279$ & $0.01156523$ \\
            $5$ & $0.00576587$ & $0.00511883$ \\
            $6$ & $0.00260819$ & $0.00223589$ \\
            $7$ & $0.00122399$ & $0.00093180$ \\
            $8$ & $0.00063960$ & $0.00031980$ \\
            \hline
        \end{tabular}
        \caption{$\gamma = \frac{13\sqrt{13}-35}{108} \approx 0.109927$, $\Gamma = 0.508672$}
        \label{tab:lp-solution-main}
    \end{subtable}
    \caption{Approximately optimal solutions to the factor-revealing LP
    with $\kappa = \sfrac{3}{2}$ and $\kmax = 8$.}
    \label{tab:lp-solution}
\end{table}

\section{Online Correlated Selection: In Detail}
\label{sec:ocs}

This section provides the formal description and analysis of the OCS used in Section~\ref{sec:edge-weighted}.
Section~\ref{sec:ocs-warmup} introduces the basics of OCS with
the proof of a $\sfrac{1}{16}$-OCS, substantiating the sketch in Section~\ref{sec:ocs-intro}.
Section~\ref{sec:ocs-main} then shows how to improve the design and analysis of
the OCS to prove Theorem~\ref{thm:ocs-main}.

\subsection{Warmup: Constructing a \texorpdfstring{$\sfrac{1}{16}$-OCS}{}}
\label{sec:ocs-warmup}

\begin{algorithm}[t]
    \caption{Online Correlated Selection (OCS)}
    \label{alg:OCS-warmup}
    \begin{algorithmic}
        \medskip
        \STATEx \textbf{State variables:}
        \begin{itemize}
            \item $\type_i \in \big\{ \chosen, \notchosen, \nonadapt \big\}$ for each ground element $i$; 
            initially, let $\type_i = \nonadapt$.
        \end{itemize}
        \STATEx \textbf{On receiving a pair of elements $i_1$ and $i_2$:}
        \begin{enumerate}
            \item With probability $\sfrac{1}{2}$, let it be a \emph{sender}:
            \begin{enumerate}
                \item Draw $\ell, m \in \{1, 2\}$ uniformly at random.
                \item Let $\type_{i_{-m}} = \nonadapt$.
                \item If $m = \ell$, let $\type_{i_m} = \chosen$; otherwise, let $\type_{i_m} = \notchosen$.
                %
                %
            \end{enumerate}
            \item Otherwise (i.e., with probability $\sfrac{1}{2}$), let it be a \emph{receiver}:
            \begin{enumerate}
                \item Draw $m \in \{1, 2\}$ uniformly at random.
                \item If $\type_{i_m} = \chosen$, let $\ell = -m$; \newline
                if $\type_{i_m} = \notchosen$, let $\ell = m$; \newline
                if $\type_{i_m} = \nonadapt$, draw $\ell \in \{1, 2\}$ uniformly at random.
                \item Let $\type_{i_1} = \type_{i_2} = \nonadapt$.
            \end{enumerate}
            \item Select $i_\ell$.
        \end{enumerate}
        \smallskip
    \end{algorithmic}
\end{algorithm}

Algorithm~\ref{alg:OCS-warmup} presents the $\sfrac{1}{16}$-OCS.
It maintains a state variable $\type_i$ for each element $i$.
If the state $\type_i$ equals $\chosen$ or $\notchosen$, it reflects the selection in the last pair involving $i$
and indicates that this information can be used in the next pair involving $i$.
If the state $\type_i$ is $\nonadapt$, it means that the past selection result
of element $i$ cannot be used to determine the selections in future pairs.

For each pair of elements $i_1$ and $i_2$ in the sequence,
the OCS first decides whether this is a \emph{sender} or a \emph{receiver} uniformly at random.
If it is a sender,
use a fresh random bit to select $i_\ell$, $\ell \in \{1, 2\}$, for this pair.
Then, draw $m \in \{1, 2\}$ uniformly at random and set $\type_{i_m}$ to reflect the selection in this round;
set $\type_{i_{-m}}$ to be $\nonadapt$, where $-m$ is an abbreviation for $3 - m$.
That is, the OCS forwards the random selection in this round to subsequent rounds for only one of the two elements in the current pair, chosen uniformly at random.

If it is a receiver, on the other hand, the OCS seeks to use the previous selection result of the elements to determine its choice of $i_\ell$.
First, it draws $m \in \{1, 2\}$ uniformly at random and checks the state variable of $i_m$.
To achieve negative correlation, the OCS makes the opposite selection in this round whenever possible.
If the state is $\chosen$, indicating that $i_m$ is selected in the last pair involving it, the OCS selects $i_{-m}$ this time, and vice versa;
if the state variable equals $\nonadapt$, the OCS uses a fresh random bit to select $i_\ell$.
In either case, reset the states of $i_1$ and $i_2$ to be $\nonadapt$.

In fact, we will show a result stronger than the definition of $\sfrac{1}{16}$-OCS.

\begin{lemma}
    \label{lem:OCS}
    For any fixed sequence of pairs of elements,
    any fixed element $i$, and any integer $k \ge 0$,
    Algorithm~\ref{alg:OCS-warmup} ensures that after appearing
    in a collection of $k$ consecutive pairs, $i$ is selected at least once with probability at least $1 - 2^{-k} \cdot f_k$, where $f_k$ is defined recursively as:
    \begin{equation}
        \label{eqn:OCS-recurrence}
        f_k = 
        \begin{cases}
          1 & \text{if $k = 0, 1$ ;} \\
            f_{k-1} - \frac{1}{16} f_{k-2} & \text{if $k \ge 2$ .}
        \end{cases}
    \end{equation}
\end{lemma}

Lemma~\ref{lem:OCS} implies that Algorithm~\ref{alg:OCS-warmup}
is a $\sfrac{1}{16}$-semi-OCS by considering the subsequence
of all pairs involving element $i$ because:
\[
  f_k = f_{k-1} - \frac{1}{16} f_{k-2}
  \le \left( 1-\frac{1}{16} \right) f_{k-1}
  \le \left( 1 - \frac{1}{16} \right)^{k-1} f_1
  = \left( 1 - \frac{1}{16} \right)^{k-1}
    ~.
\]

Let $P^1 = \{i_1^1, i_2^1\}, P^2 = \{i_1^2, i_2^2\}, \dots, P^n = \{i_1^n, i_2^n\}$
be the sequence of pairs of ground elements.
We start with a graph-theoretic interpretation of the OCS algorithm.

\paragraph{Ex-ante Dependence Graph.}
Consider a graph $G^{\textrm{\em ex-ante}} = (V, E^{\textrm{\em ex-ante}})$ as follows,
which we shall refer to as the \emph{ex-ante dependence graph}. 
To make a distinction with the vertices and edges in the online matching problem,
we shall refer to the vertices and edges in the dependence graph as \emph{nodes} and \emph{arcs}, respectively.
Let there be a node for each pair of elements in the collection.
We will refer to them as $1 \le j \le n$, i.e.:
\[
  V = \big\{ j \in \mathbb{Z} : 1 \le j \le n \big\}
    ~.
\]
Further, for any fixed element $i$ in the ground set,
let there be a directed arc from $j_1$ to $j_2$ for any two consecutive pairs $j_1 < j_2$ involving $i$, i.e.:
\[
    E^{\textrm{\em ex-ante}} = \big\{ (j_1, j_2)_i : j_1 < j_2 \textrm{ s.t.\ } i \in P^{j_1}, i \in P^{j_2}, \textrm{ and } \forall j_1 < t < j_2, i \notin P^t \big\}
    ~.
\]
The subscript $i$ helps to distinguish parallel arcs
when the pairs $j_1$ and $j_2$ have the same two elements.
See Figure~\ref{fig:dependence-graphs-ex-ante} for an illustrative example
of the ex-ante dependence graph.

\begin{figure}[t]
    \centering
    \begin{subfigure}{.85\textwidth}
        \includegraphics[width=\textwidth]{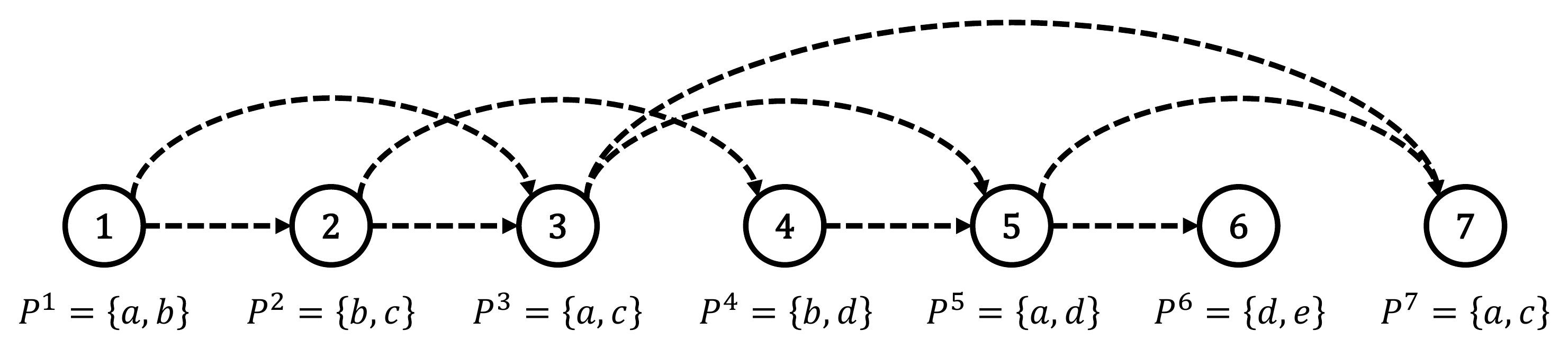}
        \caption{\emph{Ex-ante} dependence graph.}
        \label{fig:dependence-graphs-ex-ante}
    \end{subfigure}
    
    \bigskip
    
    \begin{subfigure}{.85\textwidth}
        \includegraphics[width=\textwidth]{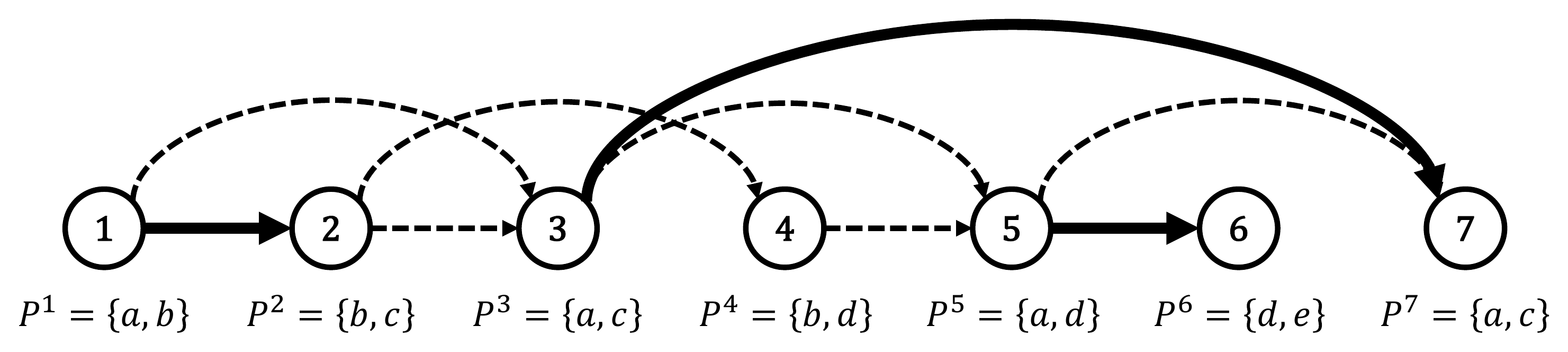}
        \caption{\emph{Ex-post} dependence graph (bold and solid edges).}
        \label{fig:dependence-graphs-ex-post}
    \end{subfigure}
    
    \bigskip
    
    \begin{subfigure}{.85\textwidth}
        \includegraphics[width=\textwidth]{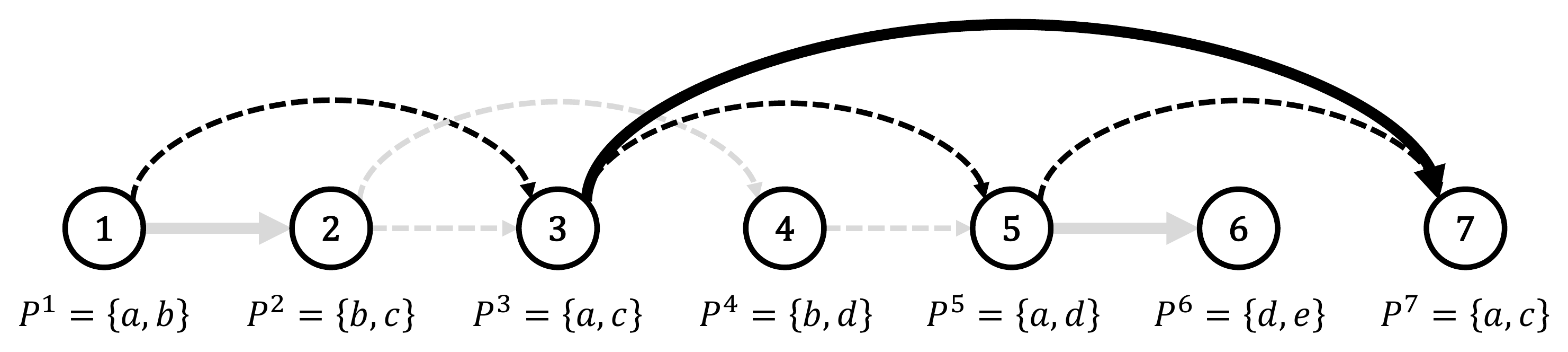}
        \caption{Dependence subgraph associated with a fixed element $a$.}
        \label{fig:dependence-graphs-fixed-candidate}
    \end{subfigure}
    \caption{Example of dependence graphs with five ground elements and a sequence of seven pairs.}
    \label{fig:dependence-graphs}
\end{figure}

Each arc in the \emph{ex-ante} dependence graph represents two pairs in the sequence
in which the OCS could use the same random bit to select oppositely.
By construction,
there are at most two outgoing arcs and at most two incoming arcs for each node.

In particular, consider any arc $(j_1, j_2)_i$ in the \emph{ex-ante} dependence graph, with $i$ being the common element.
If the randomness used by the OCS satisfies
(1) pair $j_1$ is a sender, (2) $i_m = i$ in pair $j_1$, (3) pair $j_2$ is a receiver, and (4) $i_m = i$ in pair $j_2$, the selections in the two pairs would be perfectly negatively correlated in the sense that $i$ is selected in exactly one of the two pairs.
Each of these four events happens independently with probability $\sfrac{1}{2}$.
Hence, we achieve the above perfect negative correlation with probability $\sfrac{1}{16}$.

\paragraph{Ex-post Dependence Graph.}
The \emph{ex-post} dependence graph $G^{\textrm{\em ex-post}} = (V, E^{\textrm{\em ex-post}})$
is a subgraph of the \emph{ex-ante} dependence graph that keeps the arcs corresponding
to pairs that are perfectly negatively correlated,
given the realization of whether each step is a sender or a receiver,
and the value of $m$ therein. 
Equivalently, the \emph{ex-post} dependence graph is realized as follows.
Over the randomness with which the OCS decides whether each step is a sender or a receiver,
and the values of $m$, each node in the \emph{ex-ante} dependence graph
effectively picks at most one of its incident arcs,
each with probability $\sfrac{1}{4}$;
an arc is realized in the \emph{ex-post} graph if both incident nodes choose it.
With this interpretation, we get that the \emph{ex-post} graph is a matching.
The OCS may be viewed as a randomized online algorithm that picks a matching
in the \emph{ex-ante} graph such that each arc in the \emph{ex-ante}
graph is chosen with probability lower bounded by a constant.
See Figure~\ref{fig:dependence-graphs-ex-post} for an example.

\begin{proof}[Proof of Lemma~\ref{lem:OCS}]
Let $j_1 < j_2 < \dots$ be the pairs involving a ground element $i$.
We will use the element $a$ and $k = 4$
in Figure~\ref{fig:dependence-graphs-fixed-candidate} as a running example,
where $j_1 = 1, j_2 = 3, j_3 = 5, j_4 = 7$
and the relevant arcs in the dependence graphs are
$(1, 3)_a$, $(3, 5)_a$, $(3, 7)_c$, and $(5, 7)_a$.

If at least one of the arcs among $j_1 < j_2 < \dots < j_k$ are realized in the \emph{ex-post} dependence graph, element $i$ must be selected at least once.
This is because the randomness
(related to the choice of~$\ell$ in the OCS) is perfectly negatively correlated
in the two incident nodes of the arc and thus,
$i$ is selected exactly once in these two steps.
Importantly, this is true even if the arc is not due to element $i$.
For example, given that the arc $(3, 7)_c$ is realized in
Figure~\ref{fig:dependence-graphs-fixed-candidate},
element $a$ must be selected at least once after step $7$.

On the other hand, if none of these arcs are realized,
then the random bits used in the $k$ steps $j_1 < j_2 < \dots < j_k$ are independent.
For example, consider element $a$ and $k = 3$ in Figure~\ref{fig:dependence-graphs-fixed-candidate}.
Element $a$ is selected independently with probability $\sfrac{1}{2}$
in steps $j_1 = 1$, $j_2 = 3$, and $j_3 = 5$,
given that neither $(1, 3)_a$ nor $(3, 5)_a$ is realized.

Importantly, even if some of these pairs are receivers in that the selections therein are based on the random bits realized earlier by some senders, from $i$'s point of view,
they are still independent of the randomness in the other rounds that $i$ is involved in.
For example, from $c$'s point of view
in Figure~\ref{fig:dependence-graphs-fixed-candidate},
even though the selection in step $2$ is determined by the selection
in step $1$, it is independent of the selections
in steps $3$ and $7$, which involve $c$.

Putting this together, the probability that $i$ is \emph{never selected}
after steps $j_1 < j_2 < \dots < j_k$ is equal to the probability that
(1) none of the arcs among these steps is realized,
times the probability that
(2) none of the $k$ independent random selections pick $i$.
This follows from the law of total probability.
The latter quantity equals $2^{-k}$, so
it remains to analyze the former.
We shall upper bound it by the probability that none of the arcs
$(j_1, j_2)_i, (j_2, j_3)_i, \dots, (j_{k-1}, j_k)_i$ is realized.
We shall omit the subscript $i$ in the rest of the proof for brevity.
Denote this event as $F_k$ and its probability as $f_k$.
    
Trivially, we have $f_0 = f_1 = 1$.
To prove that the recurrence in Eqn.~\eqref{eqn:OCS-recurrence} governs $f_k$, further divide event $F_k$ into two subevents.
Let $A_k$ the event that none of the arcs $(j_1, j_2), (j_2, j_3), \dots, (j_{k-1}, j_k)$ is realized, and node $j_k$ picks arc $(j_k, j_{k+1})$ in realizing the \emph{ex-post} dependence graph.
Let $B_k$ be the event that none of the arcs is realized and node $j_k$ does not pick arc $(j_k, j_{k+1})$.
Let $a_k$ and $b_k$ be the probability of $A_k$ and $B_k$, respectively.
We have that $A_k$ and $B_k$ form a partition of $F_k$, and thus:
\[
    f_k = a_k + b_k
    ~.
\]
    
    If node $j_k$ picks arc $(j_k, j_{k+1})$, which happens with probability $\sfrac{1}{4}$, arc $(j_{k-1}, j_k)$ is not realized by definition, regardless of the remaining randomness.
    Therefore, conditioned on the choice of $j_k$, subevent $A_k$ happens if and only if the choices made by steps $j_1, j_2, \dots j_{k-1}$ is such that none of $(j_1, j_2), \dots, (j_{k-2}, j_{k-1})$ is realized, i.e., when $F_{k-1}$ happens.
    That is:
    \[
        a_k = \frac{1}{4} f_{k-1}
        ~.
    \]
    
    On the other hand, if pair $j_k$ does not pick $(j_k, j_{k+1})$, there are two possibilities.
    The first case is when $j_k$ picks $(j_{k-1}, j_k)$, which happens with probability $\sfrac{1}{4}$.
    In this case, the choices made by $j_1, \dots, j_{k-1}$ must be such that none of $(j_1, j_2), \dots, (j_{k-2}, j_{k-1})$ is realized, and $j_{k-1}$ does \emph{not} pick $(j_{k-1}, j_k)$, i.e., $B_{k-1}$ happens.
    The second case is when $j_k$ picks neither $(j_{k-1}, j_k)$ nor $(j_k, j_{k+1})$, which happens with probability $\sfrac{1}{2}$.
    In this case, the choices made by $j_1, \dots, j_{k-1}$ must be such that none of $(j_1, j_2), \dots, (j_{k-2}, j_{k-1})$ is realized, i.e., $F_{k-1}$ happens.
    Putting this together, we have:
    \[
        b_k = \frac{1}{4} b_{k-1} + \frac{1}{2} f_{k-1}
        ~.
    \]
    
    \noindent
    Eliminating $a_k$'s and $b_k$'s with the above three equations, we get the recurrence in Eqn.~\eqref{eqn:OCS-recurrence}.
\end{proof}

The same analysis generalizes to prove a stronger result, which implies a $\sfrac{1}{16}$-OCS.

\begin{lemma}
    \label{lem:OCS-generalized}
    For any fixed sequence of pairs of elements, any fixed element $i$,
    and any disjoint collections of $k_1, k_2, \dots, k_m$ consecutive pairs involving $i$,
    Algorithm~\ref{alg:OCS-warmup} ensures that $i$ is selected in at least one of these pairs with probability at least:
    \[
        1 - \prod_{\ell=1}^m 2^{-k_\ell} \cdot f_{k_\ell} 
        ~.
    \]
\end{lemma}

\begin{proof}
Let $j_1^\ell < j_2^\ell < \dots < j_{k_\ell}^\ell$ be the $\ell$-th subsequence
of consecutive pairs involving element $i$, for any $1 \le \ell \le m$. 
The probability that $i$ is never selected is equal to (1) the probability
that none of the arcs among the steps in these collections is realized, times (2) the probability that all $\sum_{\ell=1}^m k_\ell$ random bits are against $i$.
The latter is $\prod_{\ell=1}^m 2^{-k_\ell}$.
We upper bound the former with the probability that for any $1 \le \ell \le m$,
none of the arcs $(j_1^\ell, j_2^\ell), \dots, (j_{k_\ell-1}^\ell j_{k_\ell}^\ell)$ is realized.
Finally, observe that the events are independent for different collections
$\ell$ because the event concerning each collection only relies on the
randomness of the nodes in the collection.
Hence, it is at most $\prod_{\ell=1}^m f_{k_\ell}$.
\end{proof}

\subsection{Optimizing the OCS: Proof of Theorem~\ref{thm:ocs-main}}
\label{sec:ocs-main}


Similar to the warmup algorithm,
we will realize the \emph{ex-post} dependence graph by letting each node
be either a \emph{sender} or a \emph{receiver} independently and randomly.
The probability of letting a node be a sender, denoted as $p$,
will be optimized later.

A sender uses a fresh random bit to select an element from the corresponding pair.
Further, it randomly picks an out-arc in the \emph{ex-post} graph and sends its selection along the out-arc.
Although the out-neighbors and out-arcs have yet to arrive,
we can refer to them as the one due to the first and second element in the current pair, respectively.
This is identical to the warmup case.

A receiver, on the other hand, adapts to the information it receives and makes the opposite selection.
The improved OCS proactively checks both in-arcs of a receiver;
in contrast, the warmup algorithm checks only one randomly chosen in-arc.
Concretely, check both in-arcs in the \emph{ex-ante} graph to see if any in-neighbor is a sender who picks the arc between them.
If both in-neighbors are senders and both pick the corresponding arcs, choose one randomly.
Add the arc to the \emph{ex-post} dependence graph.
Suppose a receiver $j$ receives the selection by a sender $j'$ sent along arc $(j', j)_i$.
Then, select $i$ in round $j$ if it is not selected in round $j'$, and vice versa.

See Algorithm~\ref{alg:improved-OCS} for a formal definition of the improved OCS.

\begin{algorithm}[t!]
    \caption{Improved Online Correlated Selection}
    \label{alg:improved-OCS}
    \begin{algorithmic}
        \medskip
        \STATEx \textbf{Parameter:}
        \begin{itemize}
            \item $p$ : probability that a node is a $\sender$.
        \end{itemize}
        \STATEx \textbf{State variables:}
        \begin{itemize}
            \item $G^{\text{\em ex-ante}} = (V, E^{\exante})$ : \emph{ex-ante} dependence graph; initially, $V = E^{\exante} = \emptyset$.
            \item $G^{\text{\em ex-post}} = (V, E^{\expost})$ : \emph{ex-post} dependence graph; initially, $V = E^{\expost} = \emptyset$.
            \item $\tau_j \in \big\{ \sender, \receiver \big\}$ for any $j \in V$.
        \end{itemize}
        \STATEx \textbf{On receiving a pair $j$ of elements $i_1$ and $i_2$:}
        \begin{enumerate}
            \item Add $j$ to $V$.
            \item For $k \in \{1, 2\}$, let $j_k$ be the last pair which involves $i_k$;
                add an arc $(j_k, j)_{i_k}$ to $E^{\exante}$.
            %
            %
            \item With probability $p$, let it be a \sender:
            \begin{enumerate}
                \item Let $i^* = i_1$ or $i_2$, each with probability $\sfrac{1}{2}$.
                \item Pick an out-arc randomly.
            \end{enumerate}
        \item Otherwise, i.e., with probability $1-p$, let it be a \receiver:
            \begin{enumerate}
                \item Pick a $j_m$, $m \in \{ 1, 2 \}$, that is a $\sender$ who picks arc $(j_m, j)_{i_m}$ (break ties randomly):
                    \label{step:improved-OCS-tie}
                \begin{enumerate}
                    \item Add arc $(j_m, j)_{i_m}$ to $E^{\expost}$.
                    \item Let $i^* = i_m$ if $i_m$ is not selected in round $j_m$, let $i^* = i_{-m}$ otherwise.
                \end{enumerate}
                \item Otherwise, let $i^* = i_1$ or $i_2$, each with probability $\sfrac{1}{2}$.
                    \label{step:improved-OCS-receiver-choice}
            \end{enumerate}
            \item Select $i^*$.
        \end{enumerate}
        \smallskip
    \end{algorithmic}
\end{algorithm}

\begin{lemma}
    \label{lem:improved-OCS}
    For any fixed sequence of pairs of elements,
    any fixed element $i$, and any disjoint subsequences of
    $k_1, k_2, \dots, k_m$ consecutive pairs involving $i$,
    Algorithm~\ref{alg:improved-OCS} ensures that $i$ is selected in at
    least one of these pairs with probability at least:
    \[
        1 - \prod_{\ell = 1}^m 2^{-k_\ell} \cdot g_{k_\ell} 
        ~,
    \]
    where $g_k$ is defined recursively as follows:
    \begin{equation}
        \label{eqn:improved-OCS-recurrence}
        g_k = \begin{cases}
          1 & \text{if $k = 0, 1$} ~; \\
            g_{k-1} - \frac{1}{8} p \left( 1-p \right) \left( 4 - p \right) \cdot g_{k-2} & \text{if $k \ge 2$} ~.
        \end{cases}
    \end{equation}
\end{lemma}

\begin{corollary}
Algorithm~\ref{alg:improved-OCS} is a $\frac{1}{8} p ( 1-p ) ( 4-p )$-OCS.
\end{corollary}

To prove Theorem~\ref{thm:ocs-main}, let $p = \frac{5-\sqrt{13}}{3}$ to maximize $\frac{1}{8} p \big( 1-p \big) \big( 4-p \big) = \frac{13\sqrt{13}-35}{108} > 0.1099$.

\subsubsection*{Proof of Lemma~\ref{lem:improved-OCS}}

Let $j^\ell_1 < j^\ell_2 < \dots < j^\ell_{k_\ell}$, $1 \le \ell \le m$,
be the subsequences of consecutive pairs that involve element~$i$.
The algorithm uses two kinds of independent random bits.
The first kind is used to realized the \emph{ex-post} dependence graph, i.e., the random type of each pair, the random out-arc chosen by each sender, and the random in-neighbor of each receiver in the tie-breaking case.
The second kind is the random selections by senders,
and by receivers that fail to receive the selection of any sender.
Importantly, the two kinds of randomness are independent.

Similar to the warmup case, we are interested in the event that there is no arc
among these pairs in the \emph{ex-post} dependence graph:
\[
    F = \left\{ \textrm{nodes $j^\ell_s$,
        for $1 \le \ell \le m$ and $1 \le s \le k_\ell$,
        are disjoint in $G^{\expost}$} \right\}
    ~.
\]

If there is an arc between two pairs in the collections in the \emph{ex-post} dependence graph, $i$ is selected in exactly one of the two pairs.
Otherwise, the selections in these pairs are independent.
Hence, the probability that $i$ is never selected is equal to the product of
(1) the probability that the $\sum_{\ell = 1}^m k_\ell$ nodes in the collections are disjoint in the \emph{ex-post} dependence graph, and (2) none of the $\sum_{\ell = 1}^m k_\ell$ independent random selections picks $i$.
This follows from the law of total probability.
The former quantity is $\Pr(F)$, and
the latter is equal to $2^{-\sum_{\ell = 1}^m k_\ell}$.
Putting this together, it equals:
\[
  2^{- \sum_{1 \le \ell \le m} k_\ell} \cdot \Pr(F)
    ~.
\]

Therefore, it remains to show that:
\begin{equation}
    \label{eqn:improved-OCS-cc-bound}
    \Pr(F) \le \prod_{\ell = 1}^m g_{k_\ell} 
    ~.
\end{equation}

Which arcs are we concerned about in the event $F$?
Since these are subsequences of consecutive pairs involving element $i$,
the arcs of the form $(j^\ell_s, j^\ell_{s+1})_i$
always exist in the \emph{ex-ante} dependence graph.
To characterize whether some of these arcs are realized in the
\emph{ex-post} graphs, we need to further consider another set of arcs as follows.

For any $1 \le \ell \le m$, consider the in-arcs of nodes
$j^\ell_1 < \dots < j^\ell_{k_\ell}$ in the \emph{ex-ante} dependence graph
due to the element other than $i$.
Let them be $(\hat{j}^\ell_s, j^\ell_s)$, for $1 \le s \le k_\ell$.
We omit the subscript that denotes the common element in the two nodes,
with the understanding that they are due to the element other than $i$ in the round of $j^\ell_s$.
Then, an arc $(j^\ell_s, j^\ell_{s+1})_i$ is realized in the \emph{ex-post} graph if:
\begin{enumerate}
    \item Node $j^\ell_s$ is a sender that picks arc $(j^\ell_s, j^\ell_{s+1})_i$;
    \item Node $j^\ell_{s+1}$ is a receiver;
    \item Either node $\hat{j}^\ell_{s+1}$ is a receiver,
      or it is a sender but does not choose arc
      $(\hat{j}^\ell_{s+1}, j^\ell_{s+1})$,
      or the tie-breaking by node $j^\ell_{s+1}$ is in favor of $j^\ell_s$.
\end{enumerate}

\paragraph{Binding Case.}
First, suppose all $\hat{j}^\ell_s$'s exist, and the $j^\ell_s$'s
and $\hat{j}^\ell_s$'s are all distinct.
It is relatively easy to analyze because in this case it suffices to
consider arcs of the form $(j^\ell_s, j^\ell_{s+1})_i$,
and different subsequences of consecutive pairs depend on disjoint sets of
random bits and therefore may be analyze separately.
This turns out to be the binding case of the analysis.
We will analyze the binding case in
Lemma~\ref{lem:improved-OCS-regular} and show in
Lemma~\ref{lem:improved-OCS-irregular} that this is the worst-case scenario
that maximizes $\Pr(F)$.

\begin{lemma}
    \label{lem:improved-OCS-regular}
    In the binding case, the probability of event $F$ is: 
    \[
      \Pr(F) = \prod_{\ell = 1}^m g_{k_\ell} 
        ~.
    \]
\end{lemma}

\begin{proof}
    We start by formalizing the aforementioned implications of the assumption that all $j^\ell_s$'s and $\hat{j}^\ell_s$'s are distinct.
    First, two pairs in the collections are connected if and only if they are consecutive pairs in the same collection, e.g., $j^\ell_{s-1}$ and $j^\ell_s$, and arc $(j^\ell_{s-1}, j^\ell_s)_i$ is realized.
    A pair $j^\ell_s$ with $s > 1$ cannot be the receiver of a sender other than $j^\ell_{s-1}$ in the collections because $\hat{j}^\ell_s$'s are not in the collections by the assumption. 
    Second, the realization of these arcs in different collections are independent.
    The realization of arcs of the form $(j^\ell_{s-1}, j^\ell_s)_i$, for any fixed collection $\ell$, depends only on the realization of first kind of randomness related to nodes with superscript $\ell$, i.e., $j^\ell_s$'s and $\hat{j}^\ell_s$'s.

    Next, we focus on a fixed subsequence $\ell$
    and analyze the probability that no arc of the form $(j^\ell_s, j^\ell_{s+1})_i$, for $1 \le s < k_\ell$, is realized.
    To simplify notation,
    we omit the superscripts and subscripts~$\ell$ and write $j_1 < j_2 < \dots < j_k$ and $\hat{j}_2, \hat{j}_3, \dots, \hat{j}_k$.
    Let $G_k$ denote this event and $g_k$ be its probability.
    Trivially, we have $g_0 = g_1 = 1$.
    It remains to show that $g_k$ follows the recurrence in Eqn.~\eqref{eqn:improved-OCS-recurrence}.
    
    We will do so by further considering an auxiliary subevent $A_k$, which requires not only $G_k$ to happen, but also $j_k$ to be a sender who picks the out-arc due to $i$.
    Let $a_k$ denote its probability.

    \paragraph{Auxiliary Event.}
    If $j_k$ is a sender who picks the out-arc due to $i$, which happens with probability~$\sfrac{p}{2}$, arc $(j_{k-1}, j_k)_i$ would not be realized regardless of the randomness of the other nodes in the collection.
    Therefore, under this condition, event $A_k$ reduces to event $G_{k-1}$. 
    \[
        a_k = \frac{p}{2} \cdot g_{k-1}
        ~.
    \]

    \paragraph{Main Event.}
    If $j_k$ is a sender, which happens with probability $p$, arc $(j_{k-1}, j_k)_i$ would not be realized regardless of the randomness of the other nodes in the collection.
    Therefore, under this condition, event $G_k$ reduces to event $G_{k-1}$. 
    The contribution of this part to the probability of $G_k$ is:
    \[
        p \cdot g_{k-1}
        ~.
    \]

    If $j_k$ is a receiver (probability $1-p$), but $\hat{j}_k$ is a sender who picks arc $(\hat{j}_k, j_k)$ (probability $\sfrac{p}{2}$), and the tie-breaking at $j_k$ is in favor of $\hat{j}_k$ (probability $\sfrac{1}{2}$), we still have that arc $(j_{k-1}, j_k)_i$ cannot be realized regardless of the randomness of the other nodes.
    The contribution of this part to the probability of $G_k$ is:
    \[
        \frac{p(1-p)}{4} g_{k-1}
        ~.
    \]

    Otherwise, $j_{k-1}$ must not be a sender who picks arc $(j_{k-1}, j_k)_i$, or else arc $(j_{k-1}, j_k)_i$ would be realized.
    Therefore, conditioned on being in this case, events $G_k$ reduces to event $G_{k-1} \setminus A_{k-1}$. 
    The contribution of this part to the probability of $G_k$ is:
    \[
        (1-p)(1-\frac{p}{4}) \left( g_{k-1} - a_{k-1} \right)
        ~.
    \]

    Putting everything together, we have:
    \[
        g_k 
        = g_{k-1} - (1-p) \left(1-\frac{p}{4}\right) a_{k-1}
        ~.
    \]

    \noindent
    Eliminating $a_k$'s by combining the two equations, we get the recurrence in Eqn.~\eqref{eqn:improved-OCS-recurrence}.
\end{proof}

\begin{lemma}
    \label{lem:improved-OCS-irregular}
    The probability of event $F$ is maximized in the binding case.
\end{lemma}

\begin{proof}
    Here are the possible violations of the conditions of the regular case: 
    \begin{enumerate}
        \item Some arc $(\hat{j}^\ell_s, j^\ell_s)$ may not exist, i.e., the element other than $i$ in pair $j^\ell_s$ has its first appearance in pair $j^\ell_s$.
        \item There may be $\ell, \ell', s, s'$ such that $\hat{j}^\ell_s = j^{\ell'}_{s'}$, i.e., the element other than $i$ in pair $j^\ell_s$ is also a element in pair $j^{\ell'}_{s'}$, and in no other pairs in between.
        \item There may be $\ell, \ell', s, s'$ such that $\hat{j}^\ell_s = \hat{j}^{\ell'}_{s'}$.
    \end{enumerate}

    We use a coupling argument to compare the probability of event $F$ in a general case, potentially with some of the above violations, with the probability in the binding case.

    \paragraph{Type-1 Violation.}
    Consider an instance almost identical to the one at hand, except we introduce a new node $\hat{j}_s^\ell$ for such a violation.
    For example, let pair $\hat{j}_s^\ell$ be at the beginning of the sequence, and let it contain the element other than $i$ in pair $j_s^\ell$ and a new dummy element that does not appear elsewhere.
    Further, couple the two instances by letting the common nodes realize identical random bits, and letting the new node draw fresh random bits.
    We claim that whenever event~$F$ happens in the original instance, it also happens in the new instance.
    If arc $(\hat{j}_s^\ell, j_s^\ell)$ is not realized, the rest of the arcs are realized identically in the two cases.
    Otherwise, having arc $(\hat{j}_s^\ell, j_s^\ell)$ may preclude arc $(j_{s-1}^\ell, j_s^\ell)_i$ from being realized, making event $F$ more likely to happen in the new instance.

    \paragraph{Type-2 Violation.}
    Consider an instance almost identical to the one at hand, except we introduce a new node $\hat{j}_s^\ell \ne j_{s'}^{\ell'}$ for such a violation.
    For example, let $\hat{j}_s^\ell$ be a pair arriving after $j_{s'}^{\ell'}$ and before $j_s^\ell$ which involves the element other than $i$ in these two pairs and a new dummy element that does not appear elsewhere.
    Further, couple the two instances by letting the common nodes realize identical random bits, and letting the new nodes draw fresh random bits.
    We claim that whenever event~$F$ happens in the original instance, it also happens in the new instance.
    Since $F$ happens in the original instance, arc $(j_{s'}^{\ell'}, j_s^\ell)$ is \emph{not} realized.
    If further arc $(\hat{j}_s^\ell, j_s^\ell)$ is not realized, the rest of the arcs are realized identically in the two cases.
    Otherwise, having arc $(\hat{j}_s^\ell, j_s^\ell)$ may preclude arc $(j_{s-1}^\ell, j_s^\ell)_i$ from being realized, making event $F$ more likely to happen in the new instance.

    \paragraph{Type-3 Violation.}
    Consider an instance almost identical to the one at hand, except we introduce a new node $\hat{j}_s^\ell \ne \hat{j}_{s'}^{\ell'}$ for such a violation.
    For example, let $\hat{j}_s^\ell$ be a pair arriving right before $j_s^\ell$ which involves the element other than $i$ in pair $j_s^\ell$ and a new dummy element that does not appear elsewhere.
    To avoid confusion in the following discussion, let $\hat{j}$ be the node in the type-3 violation the original instance, and let $\hat{j}_s^\ell \ne \hat{j}_{s'}^{\ell'}$ be the nodes in the new instance.
    Further, couple the two instances by letting nodes other than $\hat{j}$, $\hat{j}_s^\ell$, and $\hat{j}_{s'}^{\ell'}$ realize identical random bits.
    To define the coupling for these three nodes, we need some notations.
    We say that node $j_s^\ell$ \emph{needs help} if node $j_{s-1}^\ell$ is a sender who picks arc $(j_{s-1}^\ell, j_s^\ell)_i$, and if node $j_s^\ell$ is a receiver who breaks tie against $j_{s-1}^\ell$.
    For $F$ to happens in this case, $\hat{j}_s^\ell$ must be a sender who picks arc $(\hat{j}_s^\ell, j_s^\ell)$.
    Define similarly for node $j_{s'}^{\ell'}$.
    If $j_s^\ell$ needs help but $j_{s'}^{\ell'}$ does not, let $j^*$ and $j_s^\ell$ realize identical random bits, and let $\hat{j}_{s'}^{\ell'}$ draw fresh random bits;
    and vice versa.
    Otherwise, i.e., if none or both of them need help, let them have independent random bits.
    Then, when at most one of them needs help, event $F$ happens in the original instance if and only if it happens in the new instance, since the realization of the relevant arcs are identical.
    If both need help, on the other hand, $F$ cannot happen in the original instance because at least one of $(j_{s-1}^\ell, j_{s}^\ell)_i$ or $(j_{s'-1}^{\ell'}, j_{s'}^{\ell'})_i$ would be realized.
\end{proof}

\section{Conclusion}
\label{sec:conclusion}

This paper presents an online primal-dual algorithm for the edge-weighted
bipartite matching problem 
that is $0.5086$-competitive,
resolving a long-standing open problem in the study of online algorithms.
In particular, this work merges and refines the results of
Fahrbach and Zadimoghaddam~\cite{zadimoghaddam2017online}
and Huang and Tao~\cite{huang2019unweighted,huang2019weighted}
to give a simpler algorithm under the online primal-dual framework.
Our work initiates the study of \emph{online correlated selection}, a key
algorithmic ingredient that quantifies the level of negative correlation in
online assignment problems, and
we believe this technique will find further applications in
other online problems.

Using independent random bits to make selections yields a $0$-OCS (no negative
correlation),
and using an imaginary $1$-OCS with perfect negative correlation results in
inconsistent assignments.
Therefore, we aim to design an online matching algorithm
that uses partial negative correlation.
We start by constructing a $\sfrac{1}{16}$-OCS,
and then we optimize this subroutine to obtain a $0.109927$-OCS.
Designing a $\gamma$-OCS with the largest possible $\gamma \in [0,1]$ is an
interesting open problem on its own,
and would directly lead to an improved competitive
ratio for the edge-weighted online bipartite matching problem.
However, even if a 1-OCS existed, the best competitive ratio that can
be achieved using this approach is at most $\sfrac{5}{9}$, as
shown in Appendix~\ref{app:hardness}.
Thus, we need fundamentally new ideas in order
to come closer to the optimal
$1-\sfrac{1}{e}$ ratio in the unweighted and vertex-weighted cases.
One potential approach is to consider an OCS that allows for more than two
candidates in each round, which we call a multiway OCS.
We leave this as another interesting open problem for future works.

\bibliographystyle{alpha}
\bibliography{reference}

\newpage
\appendix

\section{Unweighted Online Matching}
\label{app:unweighted}

This section shows that the two-choice greedy algorithm
is strictly better than
$\sfrac{1}{2}$-competitive when combined with an
OCS in the randomized rounds to ensure partial negative correlation.

\begin{theorem}
\label{thm:main-unweighted}
The two-choice greedy algorithm
with the randomized rounds that use a $0.1099$-OCS
is at least $0.508$-competitive for unweighted online bipartite matching.
\end{theorem}

\begin{proof}
  In the unweighted case, it suffices to consider a single weight-level $w = 1$.
  Thus, for each offline vertex $i$,
  we write $k_i = k_i(1)$ for brevity.
  We will maintain $\bar{x}_i = 1 - 2^{-k_i} \cdot g_{k_i}$
  for each offline vertex~$i$,
  which according to Lemma~\ref{lem:improved-OCS}
  lower bounds the probability that $i$ is matched.
  Correspondingly, we maintain the following
  lower bound on the primal objective:
  \[
      \bar{\primal} = \sum_{i \in L} \bar{x}_i
      ~.
  \]

  To prove the stated competitive ratio, it suffices to explain how to maintain
  a dual assignment such that (1) the dual objective equals the lower
  bound of the primal objective, i.e., $\dual = \bar{\primal}$, and (2) it is
  approximately feasible up to a $\Gamma$ factor,
  i.e., $\alpha_i + \beta_j \ge \Gamma$ for every edge $(i, j) \in E$.

    \subsubsection*{Dual Updates}
    The dual updates are based on a solution to a finite version
    of the following LP.
    All of the solution values are presented in Table~\ref{tab:lp-solution-unweighted}
    at the end of this section.
    The constraints below are simpler than
    in the more general edge-weighted case,
    but the competitive ratio we achieve is almost the same.
    Note in this LP that $\Delta \alpha(k) = \alpha(k+1) - \alpha(k)$ denotes
    the forward difference operator.
    
    \begin{lemma}
        \label{lem:ratio-lp}
        The optimal value of the LP below is at least $0.508$:
        \begin{align}
            \text{\rm maximize} \quad & \Gamma \notag \\[1ex]
            \text{\rm subject to} \quad & \Delta \alpha(k) + \Delta \beta(k) \le 2^{-k} \cdot g_k - 2^{-k-1} \cdot g_{k+1} & \forall k \ge 0 \label{eqn:lp-gain-split} \\
            & \sum_{\ell = 0}^{k-1} \Delta\alpha(\ell) + 2 \cdot \Delta \beta(k) \ge \Gamma & \forall k \ge 0 \label{eqn:lp-dual-feasibility} \\
            & \Delta \beta(k) \ge \Delta \beta(k+1) & \forall k \ge 0 \label{eqn:lp-monotone-beta} \\[2ex]
            & \Delta\alpha(k), \Delta \beta(k) \ge 0 & \forall k \ge 0 \notag
        \end{align}

    \end{lemma}

    \noindent
    
    Consider an online vertex $j \in R$, and
    recall that $\kmin = \min_{i \in N(j)} k_i$ denotes the minimum value of~$k_i$
    among offline neighbors $i$ of vertex $j$.
    First suppose it is a randomized round.
    Recall that $i_1$ and~$i_2$ denote the two candidate offline
    vertices shortlisted in round $j$.
    Then, we have $k_i = \kmin$ for both $i \in \{i_1, i_2\}$.
    Note in the unweighted case that the algorithm would enter a deterministic round
    if there is a unique offline vertex with minimum $k_i$.
    In the primal, $\bar{x}_i$
    increases by $2^{-\kmin} \cdot g_{\kmin} - 2^{-\kmin-1} \cdot g_{\kmin+1}$ for both $i \in \{i_1, i_2\}$. 
    In the dual, increase $\alpha_i$ by $\Delta \alpha(\kmin)$ for both $i \in \{i_1, i_2\}$,
    and
    let $\beta_j = 2 \cdot \Delta \beta(\kmin)$
    where each $i \in \{i_1, i_2\}$ contributes $\Delta \beta(\kmin)$. 
    
    Next, suppose it is a deterministic round.
    Recall that $i^*$ denotes the offline vertex to which vertex $j$ is matched deterministically.
    Then, $x_{i^*}$ increases by $2^{-\kmin} \cdot g_{\kmin}$ in the primal.
    In the dual, increase $\alpha_{i^*}$ by $\sum_{\ell \ge \kmin} \Delta \alpha(\ell)$, and let $\beta_j = 2 \cdot \Delta \beta(\kmin+1)$.
    No update is needed in an unmatched round, as $\bar{\primal}$ remains the same. 
    
    \subsubsection*{Objective Comparisons}
    
    Next, we show that the increment in the dual objective $\dual$ is at most that in the lower bound of primal objective, i.e., $\bar{\primal}$.
    In a randomized round, it follows by Eqn.~\eqref{eqn:lp-gain-split}.
    In a deterministic round, it follows by a sequence of inequalities below:
    \begin{align*}
        \sum_{\ell \ge k} \Delta \alpha(\ell) + 2 \cdot \Delta \beta(k+1) 
        & 
        \le \sum_{\ell \ge k} \Delta \alpha(\ell) + \Delta \beta(k) + \Delta \beta(k+1) && \textrm{(Eqn.~\eqref{eqn:lp-monotone-beta})} \\
        & 
        \le \sum_{\ell \ge k} \big( \Delta \alpha(\ell) + \Delta \beta(\ell) \big) \\
        & 
        \le \sum_{\ell \ge k} \big( 2^{-\ell} \cdot g_\ell - 2^{-\ell-1} \cdot g_{\ell+1} \big) && \textrm{(Eqn.~\eqref{eqn:lp-gain-split})} \\[1ex]
        &
        = 2^{-k} \cdot g_{k}
        ~.
    \end{align*}

    \subsubsection*{Approximate Dual Feasibility.}
    
    We first summarize the following invariants, which follow by the definition of the dual updates.
    \begin{itemize}
        \item For any offline vertex $i \in L$, $\alpha_i = \sum_{\ell=0}^{k_i-1} \Delta \alpha(\ell)$.
        \item For any online vertex $j$, $\beta_j = 2 \cdot \Delta \beta(k)$
          if it is matched either in a randomized round to neighbors with $k_i = k$, or in a deterministic round to a neighbor with $k_i = k-1$.
    \end{itemize}
    
    For any edge $(i, j) \in E$, consider the value of $k_i$ at the time when $j$ arrives.
    If $k_i = \infty$, the value of $\alpha_i$ alone ensures approximately dual feasibility because:
    \begin{align*}
        \alpha_i 
        & 
        = \sum_{\ell \ge 0} \Delta \alpha(\ell) \\
        &
        = \lim_{k \to \infty} \sum_{\ell = 0}^{k-1} \Delta \alpha(\ell) \\[1ex]
        &
        \ge \Gamma - \lim_{k \to \infty} \beta(k) && \textrm{(Eqn.~\eqref{eqn:lp-dual-feasibility})} \\[2ex]
        &
        = \Gamma 
        ~.
        && \textrm{(Eqn.~\eqref{eqn:lp-gain-split}, whose RHS tends to $0$)}
    \end{align*}
    
    Otherwise, by the definition of the two-choice greedy algorithm,
    $j$ is either matched in a randomized round to two vertices
    with $k_{i'} \le k_i$, or in a deterministic round to a vertex with $k_{i'} < k_i$.
    In both cases, we have:
    \[
        \beta_j \ge 2 \cdot \Delta \beta(k_i)
        ~.
    \]
    
    \noindent
    Approximate dual feasibility now follows by $\alpha_i = \sum_{\ell = 0}^{k_i-1} \Delta \alpha(\ell)$ and Eqn.~\eqref{eqn:lp-dual-feasibility}.
\end{proof}

\begin{proof}[Proof of Lemma~\ref{lem:ratio-lp}]
    Consider a restricted version of the LP which is finite.
    For some positive $\kmax$, let $\Delta \alpha(k) = \Delta \beta(k) = 0$ for all $k > \kmax$.
    Then, the linear program becomes:
    \begin{align*}
        \text{\rm maximize} \quad & \Gamma \notag \\[2ex]
        \text{\rm subject to} \quad & \Delta \alpha(k) + \Delta \beta(k) \le 2^{-k} \cdot g_k - 2^{-k-1} \cdot g_{k+1} & 0 \le k \le \kmax \\[1ex]
        & \sum_{\ell = 0}^{k-1} \Delta\alpha(\ell) + 2 \cdot \Delta \beta(k) \ge \Gamma & 0 \le k \le \kmax \\
        & \sum_{\ell = 0}^{\kmax} \Delta \alpha(\ell) \ge \Gamma \\[1.5ex]
        & \beta(k) \ge \beta(k+1) & 0 \le k < \kmax \\[4ex]
        & \Delta\alpha(k), \beta(k) \ge 0 & 0 \le k \le \kmax \notag
    \end{align*}
 
    See Table~\ref{tab:lp-solution-unweighted} for an approximately optimal solution for the restricted LP with $\kmax = 8$,
    which gives a competitive ratio of $\Gamma \approx 0.508986$.
\end{proof}

\begin{table}[t]
    \centering
    \begin{tabular}{c|ccc}
        \hline
        $k$     & $g_k$     & $\Delta \alpha(k)$    & $\beta(k)$    \\
        \hline
        $0$     & $1.00000000$     & $0.24550678$             & $0.25449322$     \\
        $1$     & $1.00000000$     & $0.14574204$             & $0.13173982$     \\
        $2$     & $0.89007253$     & $0.06613120$             & $0.05886880$     \\
        $3$     & $0.78014506$   & $0.02907108$             & $0.02580320$     \\
        $4$     & $0.68230164$ & $0.01273424$             & $0.01126766$     \\
        $5$     & $0.59654227$ & $0.00559236$             & $0.00490054$     \\
        $6$     & $0.52153858$ & $0.00248248$             & $0.00210436$     \\
        $7$     & $0.45596220$ & $0.00114193$             & $0.00086312$     \\
        $8$     & $0.39863078$ & $0.00058431$             & $0.00029216$     \\
        \hline
    \end{tabular}
    \caption{An approximately optimal solution to the factor-revealing linear program
    with $\kmax = 8$ in the unweighted case.
    The competitive ratio obtained is $\Gamma \approx 0.508986$.}
    \label{tab:lp-solution-unweighted}
\end{table}

\section{Hard Instances}
\label{app:hardness}


This section presents two families of \emph{unweighted} graphs that
demonstrate some hardness results for the online matching 
algorithms considered in this paper.

\subsection{Upper Triangular Graphs}

Consider a bipartite graph with $n$ vertices on each side.
Let each online vertex $1 \le j \le n$ be incident to the offline vertices $j \le i \le n$.
Thus, the adjacency matrix
(with online vertices as rows and offline vertices as columns)
is an upper triangular matrix.
This is a standard instance for showing hardness that
dates back to Karp et al.~\cite{karp1990optimal}.

\begin{theorem}
    The two-choice greedy algorithm using independent random bits in different
    randomized rounds is only $(\sfrac{1}{2} + o(1))$-competitive.
\end{theorem}

\begin{proof}
    For ease of presentation, suppose the algorithm chooses candidates in \emph{reverse} lexicographical order.
    Consider an upper triangular graph with $n = 3^k$ for some large positive integer $k$.
    First, observe that there is a perfect matching
    where the $i$-th online vertex is matched to the $i$-th offline vertex.
    Hence, the optimal value is $n$.

    Next, consider the performance of the online algorithm.
    The first $\sfrac{n}{3} = 3^{k-1}$
    vertices are matched to the last $\sfrac{2}{3}$ fraction of the offline vertices in randomized rounds.
    That is, their correct neighbors in the perfect matching are left unmatched,
    while the other offline vertices are only half matched.
    Then, the first one third of the remaining online vertices
    (i.e., $\sfrac{1}{3} \cdot \sfrac{2n}{3} = 2 \cdot 3^{k-2}$ of them in total)
    are matched to the last $(\sfrac{2}{3})^2$ fraction of the offline vertices in randomized rounds.
    That is, their correct neighbors in the perfect matching are left matched by only half,
    while the correct neighbors of subsequent online vertices are now matched by three quarters.
    The argument goes on recursively.

    Therefore, omitting a lower order term due to the last
    $2^k = n^{\log_3 2}$ vertices on both sides, the expected size of the matching is:
    \begin{align*}
        \left( 1 \cdot \frac{1}{3} + \frac{1}{2} \cdot \frac{2}{9} + \cdots + \left(\frac{1}{2}\right)^k \cdot \frac{2^k}{3^{k+1}} + \cdots \right) n 
        & 
        = \left( \frac{1}{3} + \frac{1}{9} + \cdots + \frac{1}{3^{k+1}} + \cdots \right) n \\
        & 
        = \frac{n}{2}
        ~.
    \end{align*}

    \noindent
    Hence, the two-choice greedy algorithm is at best $(\sfrac{1}{2} + o(1))$-competitive.
\end{proof}

\begin{theorem}
    The imaginary two-choice greedy algorithm with perfect negative correlation across
    different randomized rounds is only $\sfrac{5}{9}$-competitive.
\end{theorem}

\begin{proof}
    For ease of presentation, suppose the algorithm chooses candidates in \emph{reverse} lexicographical order.
    Consider an upper triangular graph with $n = 9$.
    There are nine vertices on each side, denoted as
    $i_1, i_2, \dots, i_9$ and $j_1, j_2, \dots, j_9$, and a perfect matching with $i_k$ matched to $j_k$ for $k = 1, 2, \dots, 9$.
    The first three online vertices, $j_1$, $j_2$, and $j_3$, are connected to all offline vertices.
    After their arrivals, $i_1$, $i_2$, and $i_3$ are unmatched while the remaining six offline vertices are matched by half.
    Then, the next two online vertices, $j_4$ and $j_5$, are connected to the last six
    offline vertices, i.e., $j_4$ to $j_9$.
    After their arrival, $i_4$ and $i_5$ remain matched by half, while $i_6$ to $i_9$ are fully matched. 
    Therefore, the algorithm finds a matching of size $\sfrac{1}{2} \cdot 2 + 1 \cdot 4 = 5$ 
    in expectation, but the optimal matching has size $9$. 
    The competitive ratio is $\sfrac{5}{9}$, which matches the lower
    bound that we want to show.
\end{proof}

\subsection{Erdös--Rényi Upper Triangular Graphs}

Consider the following random bipartite graph that has $n$ vertices on each side.
Each online vertex $1 \le j \le n$ is incident to the offline vertex $i = j$ with certainty,
and each offline vertex $j < i \le n$ is adjacent to $j$
independently with probability $p$, where $0 < p < 1$ is a parameter to
be determined.

By considering the Erdös--Rényi variant of upper triangular graphs
instead of the original ones,
we ensure that with high probability any fixed online vertex is paired
with different offline vertices in its randomized round.
This is effectively the worst-case scenario in the analysis of the OCS algorithm
in Section~\ref{sec:ocs}.
Letting $n = 2^{13}$ and $p = 2^{-6}$,
an empirical evaluation shows that our analysis for the combination
of a two-choice greedy algorithm and with an OCS is nearly optimal.

\begin{theorem}
    The competitive ratio of the two-choice greedy algorithm with the OCS in Algorithm~\ref{alg:OCS-warmup} is at most $0.5057$-competitive.
\end{theorem}

\begin{theorem}
    The competitive ratio of the two-choice greedy algorithm with the OCS in Algorithm~\ref{alg:improved-OCS} is at most $0.51$-competitive.
\end{theorem}

\section{Connections to the Original Algorithm}
\label{sec:origin}


\newcommand{\Color}{\textsf{active}}
\newcommand{\Index}{\textsf{index}}
\newcommand{\Partner}{\textsf{partner}}
\newcommand{\Mark}{\textsf{priority}}
\newcommand{\green}{\textsf{false}}
\newcommand{\blue}{\textsf{true}}

\newcommand{\Gain}{\textsf{gain}}
\newcommand{\AdaptiveGain}{\textsf{adaptive\_gain}}

This section explains the connections between
the online primal-dual algorithm in this paper
to the original algorithm by Fahrbach and Zadimoghaddam~\cite{zadimoghaddam2017online}.
We start by briefly describing the algorithm in Algorithm~\ref{alg:origin},
with minor modifications to make it consistent with the notations in this paper.
Next, we simplify the algorithm by considering the special case of
unweighted online matching and present the result as Algorithm~\ref{alg:origin-unweighted}.
Finally, we explain how
Algorithm~\ref{alg:origin-unweighted} in the unweighted case
is effectively a two-choice greedy algorithm that implicitly uses the
warmup $\sfrac{1}{16}$-OCS from Section~\ref{sec:ocs}.

\subsection{Original Algorithm}

The algorithm uses two parameters $\varepsilon, \delta > 0$
that are later optimized in the analysis in~\cite{zadimoghaddam2017online}.
These parameters are not necessary for explaining the connections between the
two algorithms, so we omit their values.
For each offline vertex $i \in L$, 
it maintains the following state variables that express the behavior of 
the last randomized round involving $i$.
First, it maintains a boolean variable $\Color(i)$ that indicates
whether the realization of the last randomized selection involving vertex $i$
can be adaptively used in the next randomized round in which $i$ is involved.
The goal here is to introduce negative correlation
in the same way that the $\sfrac{1}{16}$-OCS does.
Each offline vertex also maintains two state variables about
the last randomized selection involving $i$:
the corresponding online vertex $\Index(i)$,
and the other offline vertex $\Partner(i)$ in the last randomized round.
Finally, the realization of the last randomized selection
is stored as $\Mark(i)$.
Informally in the notation of this paper,
$\Mark(i) = 0$ corresponds to the case when the
last randomized round involving $i$ is a receiver;
otherwise (i.e., the sender case)
$\Mark(i)$ is $1$ if $i$ is selected the last time, and is $2$ if $i$ is not selected.

\begin{algorithm}
    \caption{Original $0.501$-competitive algorithm for edge-weighted bipartite matching in
    \cite{zadimoghaddam2017online}.}
    \label{alg:origin}
    \begin{algorithmic}
        \STATE \textbf{parameters:~} $\varepsilon$, $\delta$
        \STATE \textbf{state variables:~}
            $\forall i \in L$,
            $\Color(i) \gets \green$,
            $\Index(i) \gets 0$,
            $\Partner(i)\gets 0$,
            $\Mark(i) \gets 0$,
            $S(i) \gets 0$
        \FORALL{online vertex $j$}
        \STATE $M_j \gets \max_{i \in L} \E[\Gain_{ij}] = \int_0^{w_{ij}} \big( 1 - y_i(w) \big) dw$
            %
            \FORALL{offline vertices $j$}
            \IF{$\Color(i) = \blue$ \textbf{and} $w_{ij} \ge w_{i,\Index(i)} - \delta M_j$}
            \STATE $\AdaptiveGain_{ij} \gets (\E[\Gain_{i,\Index(i)}]/3 - (w_{i,\Index(i)} - w_{ij})^+/3 - S(i))^+/12$\;
                \ELSE
                    \STATE $\AdaptiveGain_{ij} \gets 0$
                \ENDIF
            \ENDFOR 
            \STATE $B \gets \{i \in L : w_{ij} \ge w_{i,\Index(i)} - \delta M_j
            \text{ and }
            \E[\Gain_{ij}] + \sfrac{2}{3} \cdot \AdaptiveGain_{ij} \ge (1-\varepsilon) M_j \}$\;
            %
            \IF{$|B| \ge 2$} 
            \hspace*{\fill}
            \textbf{(case 1: there are enough candidates to exploit adaptivity)}
                \STATE pick $i_1, i_2 \in B$ with the largest $\E[\Gain_{ij}] + 2/3 \cdot \AdaptiveGain_{ij}$\;
                \STATE set $\Color(\Partner(i)) \gets \green$ for $i = i_1, i_2$\;
                \STATE set $\Color(i) \gets \blue$, $S(i) \gets 0$, $\Index(i) \gets j$ for $i = i_1, i_2$\;
                \STATE set $\Partner(i_1) \gets i_2$ and $\Partner(i_2) \gets i_1$\;
                \STATE pick $\ell \in \{1, 2\}$ with the larger $\AdaptiveGain_{i_\ell j}$ and let $-\ell$ denote $3 - \ell$\;
                \STATE draw $R \in [0, 1)$ uniformly at random\;
                \IF{$R \in [0,\sfrac{1}{3})$ or $\AdaptiveGain_{i_\ell j} = 0$}
                    \STATE \textbf{if} $\Mark_{i_\ell} = 2$ \textbf{and} $\AdaptiveGain_{i_\ell j} > 0$ \textbf{then} match $j$ to $i_\ell$\;
                    \STATE \textbf{if} $\Mark_{i_\ell} = 1$ \textbf{and} $\AdaptiveGain_{i_\ell j} > 0$ \textbf{then} match $j$ to $i_{-\ell}$\;
                    \STATE \textbf{if} $\Mark_{i_\ell} = 0$ \textbf{or} $\AdaptiveGain_{i_\ell j} = 0$ \textbf{then} match $j$ to $i_1$ or $i_2$ with equal probability
                    \STATE set $\Mark_{i_1} \gets 0$ and $\Mark_{i_2} \gets 0$\;
                \ELSE
                    \STATE \textbf{if} $R \in [\sfrac{1}{3}, \sfrac{2}{3})$ \textbf{then} assign $i$ to $i_1$ and set $\Mark_{i_1} \gets 1$, $\Mark_{i_2} \gets 2$
                    \STATE \textbf{if} $R \in [\sfrac{2}{3}, 1)$ \textbf{then} assign $i$ to $i_2$ and set $\Mark_{i_1} \gets 2$, $\Mark_{i_2} \gets 1$
                \ENDIF
            \ELSE \hspace*{\fill} \textbf{(case 2: there is no adaptivity)}
                \STATE $B' \gets \{i \in L : (w_{ij} \ge w_{i,\Index(i)} - \delta M_j) \textnormal{ and } (\E[\Gain_{ij}] \ge (1-\varepsilon)M_j) \}$\;
                \STATE $C \gets \{i \in L : (w_{ij} < w_{i,\Index(i)} - \delta M_j) \textnormal{ and } (\E[\Gain_{ij}] \ge (1-\varepsilon)M_j)
            \}$\;
                \IF{$|B' \cup C| = 1$}
                    \STATE match $j$ to $i_1 \gets \argmax_{i \in L} \E[\Gain_{ij}]$\;
                    \STATE set $S(i_1) \gets S(i_1) + M_j$\;
                \ELSE
                    \IF{$B' \ne \emptyset$}
                        \STATE $i_1 \gets$ the only advertiser in $B'$
                    \ELSE
                        \STATE $i_1 \gets \argmax_{i \in C}\E[\Gain_{ij}]$
                    \ENDIF
                    \STATE $i_2 \gets \argmax_{i \in C \setminus\{i_1\}} \E[\Gain_{ij}]$\;
                    \STATE match $j$ to $i_1$ or $i_2$ with equal probability\;
                    \STATE set $S(i_1) \gets S(i_1) + M_j/2$ and $S(i_2) \gets S(i_2) + M_j/2$
                \ENDIF
            \ENDIF 
        \ENDFOR
    \end{algorithmic}
\end{algorithm}

For each online vertex $j \in R$, the matching decision is made using
two different quality measures for the offline vertices.
For each offline vertex $i$,
$\Gain_{ij}$ denotes how much $i$'s heaviest edge weight would increase,
should $j$ be matched to $i$.
The first measure is the expectation of $\Gain_{ij}$,
which equals $\int_0^{w_{ij}} ( 1 - y_i(w) ) dw$
in the CCDF viewpoint of this paper.
It also defines $\AdaptiveGain_{ij}$,
which captures the extra value of matching $j$ to $i$
due to the ability to make adaptive decisions based on the realization of
the last randomized round involving $i$.
Informally, this corresponds to the benefit of using the OCS for
negative correlation in our primal-dual algorithm.
The formula of $\AdaptiveGain_{ij}$ is
derived from the analysis in \cite{zadimoghaddam2017online}
and its interpretation is not necessary
for understanding the connections between the two algorithms.
We refer the reader to~\cite{zadimoghaddam2017online} for a more detailed
explanation of Algorithm~\ref{alg:origin} in the general edge-weighted online
matching problem.

\subsection{Simplified and Symmetrized Algorithm for Unweighted Online Matching}

We now focus on a simplified algorithm in the special case of unweighted
online matching in order to better explain
the connections to our primal-dual algorithm in this paper.
In this setting, $w_{ij} \in \{0,1\}$ for any $i \in L$ and any $j \in R$.

\paragraph{Simplifying Case 2.}
The case when $|B| \le 1$ in Algorithm~\ref{alg:origin} can be significantly simplified in the unweighted case.
Observe that for any offline vertex $i \in L$, either $w_{ij} = 0$,
in which case we have $\E \big[ \Gain_{ij} \big] = 0$,
or $w_{ij} = 1$, in which case we have $w_{ij} > w_{i,\Index(i)}-\delta M_j$.
In other words, $C$ is always an empty set.
On the other hand,
$|B'|$ is nonempty because the offline neighbor
with the maximum value of $\E \big[ \Gain_{ij} \big]$ is always in the set.
Further observe that $B'$ must be a singleton because $B'$
is a subset of $B$, which has at most one element
since we are in the second case of the algorithm.
Putting this all together,
the algorithm always matches $i$ to the unique element in $B'$.
This corresponds to a deterministic round in our primal-dual
algorithm in this paper.


\paragraph{Simplifying Gains and Adaptive Gains.}
Recall that $x_i$ denotes the probability that an offline vertex $i$ is matched.
Thus, the expected gain $\E \big[ \Gain_{ij} \big]$
in the unweighted case equals $1 - x_i$.
Observe that the expected gain is $0$ if an offline vertex $i$ is involved in case 2
(i.e., a deterministic round).

Next, we consider the adaptive gain values.
In the unweighted case, the second term in the formula for computing
adaptive gains is always $0$ for any offline neighbor $i$ of the online vertex $j$,
because both $w_{i,\Index(i)}$ and $w_{ij}$ are $1$.
The third term, on the other hand,
equals $0$ if $i$ has never been in case 2,
and otherwise equals $\E[ \Gain_{i,\Index(i)}]$
due to the discussion above on the simplification of case~2 in the unweighted case.
In other words, the adaptive gain can be simplified as
$\E[ \Gain_{i, \Index(i)} ]/36$ for any $i$ that has not yet been matched deterministically,
and is $0$ otherwise.

\paragraph{Simplifying the Candidate Set.}
Since both the gain and the adaptive gain values
are $0$ for any offline vertex that has been deterministically matched,
they cannot appear in the candidate set $B$.
Therefore, it suffices to consider $j$'s offline neighbors that have not
yet been matched deterministically.
For such vertices, the first condition of the candidate set $B$
holds trivially because $w_{ij} = 1$ and
$w_{i, \Index(i)} - \delta M_j = 1 - \delta M_j < 1$.
In conclusion, it suffices to keep only the second condition.

\begin{algorithm}
    \caption{Simplified and symmetrized version of the original algorithm in the unweighted case.}
    \label{alg:origin-unweighted}
    \begin{algorithmic}
        \STATE \textbf{parameters:~} $\varepsilon$, $\delta$
        \STATE \textbf{state variables:~}
            $\forall i \in L$,
            $\Color(i) \gets \green$,
            $\Index(i) \gets 0$,
            $\Partner(i)\gets 0$,
            $\Mark(i) \gets 0$,
            $S(i) \gets 0$
        \FORALL{online vertices $j$}
            \STATE let $U_j$ be the set of neighbors of $j$ that have \emph{not} been matched deterministically\;
            \STATE $M_j \gets \max_{i \in U_j} \E[\Gain_{ij}] = 1 - x_i$
            %
            \FORALL{$i \in U_j$}
                \IF{$\Color(i) = \blue$}
                    \STATE $\AdaptiveGain_{ij} \gets \E[\Gain_{i,\Index(i)}] / 36$\;
                \ELSE
                    \STATE $\AdaptiveGain_{ij} \gets 0$\;
                \ENDIF
            \ENDFOR 
            \STATE $B \gets \{i \in U_j : \E[\Gain_{ij}] + \sfrac{2}{3} \cdot \AdaptiveGain_{ij} \ge (1-\varepsilon) M_j \}$\;
            %
            \IF{$|B| \ge 2$} 
            \hspace*{\fill}
            \textbf{(randomized round)}
                \STATE pick $i_1, i_2 \in B$ with the largest $\E[\Gain_{ij}] + \sfrac{2}{3} \cdot \AdaptiveGain_{ij}$\;
                \STATE set $\Color(\Partner(i)) \gets \green$ for $i = i_1, i_2$\;
                \STATE set $\Color(i) \gets \blue$, $S(i) \gets 0$, $\Index(i) \gets j$ for $i = i_1, i_2$\;
                \STATE set $\Partner(i_1) \gets i_2$ and $\Partner(i_2) \gets i_1$\;
                \STATE draw $\ell \in \{1, 2\}$ uniformly at random and let $-\ell$ denote $3 - \ell$\;
                \STATE draw $R \in [0, 1)$ uniformly at random\;
                \IF{$R \in [0,\sfrac{1}{3})$}
                    \STATE \textbf{if} $\Mark_{i_\ell} = 2$ \textbf{and} $\AdaptiveGain_{i_\ell j} > 0$ \textbf{then} match $j$ to $i_\ell$\;
                    \STATE \textbf{if} $\Mark_{i_\ell} = 1$ \textbf{and} $\AdaptiveGain_{i_\ell j} > 0$ \textbf{then} match $j$ to $i_{-\ell}$\;
                    \STATE \textbf{if} $\Mark_{i_\ell} = 0$ \textbf{or} $\AdaptiveGain_{i_\ell j} = 0$ \textbf{then} match $j$ to $i_1$ or $i_2$ with equal probability
                    \STATE set $\Mark_{i_1} \gets 0$ and $\Mark_{i_2} \gets 0$\;
                \ELSE
                    \STATE \textbf{if} $R \in [\sfrac{1}{3}, \sfrac{2}{3})$ \textbf{then} assign $i$ to $i_1$ and set $\Mark_{i_1} \gets 1$, $\Mark_{i_2} \gets 2$
                    \STATE \textbf{if} $R \in [\sfrac{2}{3}, 1)$ \textbf{then} assign $i$ to $i_2$ and set $\Mark_{i_1} \gets 2$, $\Mark_{i_2} \gets 1$
                \ENDIF
            \ELSE \hspace*{\fill} \textbf{(deterministic round)}
                \STATE match $j$ to the $i$ with the largest $\E[\Gain_{ij}]$\;
            \ENDIF 
        \ENDFOR
    \end{algorithmic}
\end{algorithm}

\paragraph{Symmetrizing the Choice of $\ell$.}
We observe that choosing $\ell$ to maximize the adaptive gain
has no significance in the analysis by~\cite{zadimoghaddam2017online}.
The analysis therein distributes the benefit of making
adaptive decision w.r.t.\
to $i_\ell$ equally between $i_1$ and $i_2$,
and for $i_{-\ell}$ it merely needs its share to be at
least half the benefit of making adaptive decision w.r.t.\ $i_{-\ell}$.
To this end, we symmetrize the choice of $\ell \in \{1, 2\}$
to be uniformly at random.
Doing so makes the connection to the algorithm in this paper more apparent.

\paragraph{Optimizing the Efficiency of Adaptivity.}
Finally, we remove the condition on the value of adaptive gain being $0$
in the if statement in the first case of the algorithm.
This is driven by the observation that the online vertex $j$
is matched randomly to $i_1$ and $i_2$ with equal probability
whenever it holds, which is identical to the else case of the if statement.
Moreover, the latter case allows us to store the realization of the
random selection to be exploited adaptively later, while the former does not.
Hence, other than making the algorithm closer to the OCS introduced
in this paper,
this technical change also improves the efficiency 
of adaptivity in the original algorithm.

\bigskip

This simplified and symmetrized version of the original algorithm
in the unweighted case
is summarized as Algorithm~\ref{alg:origin-unweighted}.

\subsection{Connections Between the Unweighted Algorithms}
%

We focus on the randomized rounds to explain the connections
to the warmup $\sfrac{1}{16}$-OCS in Section~\ref{sec:ocs}.
If $R \in [\sfrac{1}{3}, \sfrac{2}{3})$, it corresponds to a sender round in the OCS where $i_1$ is selected.
If $R \in [\sfrac{2}{3}, 1)$, it corresponds to a sender round where $i_2$ is selected.
If $R \in [0, \sfrac{1}{3})$, it corresponds to a receiver round.
The choice of $\ell$ corresponds to the choice of a random
in-arc by a receiver in the OCS.
The state variable $\Color(i)$ ensures that each sender's
selection is adaptively used by at most one receiver.
In the OCS, the sender randomly picks an out-arc as the potential receiver.
In contrast, Algorithm~\ref{alg:origin-unweighted}
deterministically picks the out-neighbor which arrives earlier;
\cite{zadimoghaddam2017online} effectively uses an amortization in the analysis to distribute the benefit between the two out-arcs.

In conclusion, aside from the different choices of constants and
the use of an amortization in the analysis instead of a symmetrized algorithm,
the $0.501-$-competitive algorithm in~\cite{zadimoghaddam2017online}
implicitly contains the ideas behind
the warmup $\sfrac{1}{16}$-OCS presented
as Algorithm~\ref{alg:OCS-warmup} in this paper.

\end{document}